\documentclass[11pt,english,american,fleqn]{article}

\usepackage[colorlinks=true,linkcolor=blue,
citecolor=orange]{hyperref}
\usepackage{amsmath,amsfonts,amssymb,amsthm,mathrsfs}
\usepackage{mathabx}
\usepackage{setspace}
\usepackage{soul}

\usepackage{enumerate}

\theoremstyle{plain}\newtheorem{theorem}{Theorem}[section]
\theoremstyle{plain}\newtheorem{lemma}[theorem]{Lemma}
\theoremstyle{plain}\newtheorem{corollary}[theorem]{Corollary}
\theoremstyle{plain}
\theoremstyle{plain}\newtheorem{proposition}[theorem]{Proposition}
\theoremstyle{definition}
\theoremstyle{remark}\newtheorem{remark}{Remark}
\theoremstyle{definition}\newtheorem{def:and:lemma}[theorem]{Definition and Lemma}

\renewcommand{\Re}{\textnormal{Re}}

\usepackage[colorinlistoftodos,prependcaption,textsize=tiny]{todonotes}

\counterwithin{remark}{section}

\newcommand{\ve}{\varepsilon}

\newcommand{\R}{\mathbb R }
\newcommand{\F}{\mathcal F}
\renewcommand{\t}[1]{\textnormal{#1}}

\newcommand{\vp}{\varphi}
\newcommand{\A}{\mathcal A}
\usepackage{dsfont}
\newcommand{\1}{\mathds 1}			
\newcommand{\id}{\textbf{1}}		
\newcommand{\h}{ h^{\rm eff}  }

\newcommand{\<}{\langle} 
\renewcommand{\>}{\rangle}
\usepackage{braket}
\renewcommand{\O}{\mathcal D}

\begin{document}

\title{\Large {\textsc{The renormalized Nelson model in the weak coupling limit}}}

	 
\author{Esteban C\'ardenas \thanks{Department of Mathematics,
		University of Texas at Austin,
		2515 Speedway,
		Austin TX, 78712, USA. E-mail: \texttt{eacardenas@utexas.edu}} \quad and\quad David Mitrouskas \thanks{Institute of Science and Technology Austria (ISTA), Am Campus 1, 3400 Klosterneuburg, Austria.\\ E-mail: \texttt{mitrouskas@ist.ac.at}}
}


\maketitle
 
\frenchspacing

\begin{spacing}{1.15} 
 
\begin{abstract}
The Nelson model describes non-relativistic particles coupled to a relativistic Bose scalar field. 
In this article, we study the renormalized version of the Nelson model with massless bosons  in Davies' weak coupling limit. 
Our main result
states that  the two-body
Coulomb potential emerges as an effective pair interaction between the particles, which arises
 from the exchange of virtual excitations of the quantum field. 
\end{abstract}


\section{Introduction}

\allowdisplaybreaks

Non-relativistic quantum electrodynamics (QED) describes the interactions between charged particles that are mediated by photons--the excitations of the quantized electromagnetic field. When particle velocities are significantly smaller than the speed of photon propagation, these interactions can be effectively approximated by a direct Coulomb interaction between the particles. The goal of this article is to provide a precise formulation of this approximation. Instead of working within the full non-relativistic QED framework, we focus on the renormalized massless Nelson model, \textit{without} infrared or ultraviolet regularization.  The Nelson model has been extensively studied in mathematical physics and provides insights into the interplay between particles and quantum fields in a simplified yet meaningful way.   Since Nelson’s pioneering work  
   \cite{Nelson64} in 1964, it has been understood that this model requires an infinite energy renormalization in the ultraviolet regime. Formally, the Nelson model for two particles coupled to a massless scalar Bose field is described by the Hamiltonian
\begin{align}\label{eq:formal}
 \sum_{i=1,2}	(  - \Delta_{x_i})  \otimes \id 
  + \id \otimes  \mu \int dk \, 
	| k | \, a_k^* 
	a_k + \mu^{1/2}  \sum_{i=1,2}
	\int
	\frac{dk}{\sqrt{  | k | }}
	\Big( e^{-ikx_i}  \otimes a_k^* + e^{ikx_i} \otimes  a_k \Big)
\end{align}
where $a_k^*$ and $a_k$ are bosonic creation and annihilation operators and $\mu>0$ is a dimensionless scaling parameter. We are interested in studying the dynamics generated by the Hamiltonian \eqref{eq:formal}
in the regime $ \mu \rightarrow \infty$, known as the \textit{weak coupling limit}. 
This limit was first studied by Davies \cite{Davies79}
and is equivalent to a scaling of small coupling 
$\lambda  = \mu^{-\frac{1}{2}}    $ 
for  heavy particles of mass $M = \mu $,
on  a rescaled time-scale $\tau =  \mu^{-1 }t $. The work of Davies was  refined by Hiroshima \cite{Hiroshima98,Hiroshima99} and Gubinelli, Hiroshima and L\"orinczi \cite{GubinelliHJ2014}. A comparison of their results with ours is given below. For an alternative interpretation of the scaling $\mu \to \infty$, we refer to \cite[Section 1]{Stefan2}.

\subsection{The renormalized massless Nelson model}

Let us now turn to the definition
of the precise model that we study. 
We  consider here   the two-particle case with Hilbert space
\begin{align}
	\mathscr H = L^2(\mathbb R^6) \otimes \mathcal F, \qquad  \mathcal F = \mathbb C  \oplus \bigoplus_{n = 1}^\infty \bigotimes^{n}_{\rm sym} L^2(\mathbb R^3),
\end{align}
where $ \F$ is the bosonic Fock space  and the Fock vacuum is denoted by $\Omega = (1, \mathbf{0})$--our analysis extends naturally to $N\ge 2$ particles.
On $\mathscr  H$
we  first  introduce the Hamiltonian with
an ultraviolet (UV) cutoff $\Lambda >0$
\begin{align}\label{eq:HN}
	H_\Lambda^{\rm N} 
 \, 	\equiv  \, 
p_1^2 \otimes \id 
 \, + \, 
  p_2^2 \otimes \id 
 \, +
 \,  \id \otimes  
 \mu T 
  \, + \,  
 \mu^{ \frac{1 }{2}}  \sum_{i=1,2}
\int\limits_{|k| \leq \Lambda } 
\frac{dk}{ \sqrt{|k|}}  
\Big( e^{-ikx_i}  \otimes a_k^* + e^{ikx_i} \otimes  a_k \Big)
\end{align}
where $T \equiv  d\Gamma(| k |)$ is the field energy
and $p_i  \equiv  - i  \nabla_{x_i }$
are the momentum operators. By the Kato--Rellich Theorem, it is straightforward to verify that $H_\Lambda^{\rm N}$ is self-adjoint on the domain of the non-interacting operator
$ D( p_1^2 +p_2^2  + \mu T) $, ensuring that $\{ e^{-it H_\Lambda^{\rm N}}\}_{t\in \mathbb R}$ defines a strongly continuous unitary group. The next proposition provides the existence of the unitary group as $\Lambda \to \infty$, after appropriate energy renormalization. This was proved in \cite{Griesemer18}, which, among other contributions, extended Nelson's work \cite{Nelson64} to the case of  massless fields.
 For alternative approaches to renormalizing the Nelson model, see \cite{GubinelliHJ2014, LS2019}. 
 For the statement, let us introduce the self-energy of  the particles as
 \begin{align}
	\label{E}
	E_\Lambda 
	 \ \equiv  \ 
	\int_{| k | \leq \Lambda }  \frac{    2 	\mu  \, d k }{ |k| (k^2 + \mu |k|) } 
 \, 	+ \, 
	e_0  
	\ = \ 
	   8 \pi  \mu 
	\ln (1+\Lambda \mu^{-1} )  \, + \,  e_0    \  . 
\end{align}
Here, 
 $e_0>0$ is   a constant  energy shift defined in      \eqref{e0} --independently of $\mu$ and $\Lambda$-- and   is added  only for  convenience. 
Notably, $E_\Lambda$  diverges logarithmically as $\Lambda \to \infty$ for fixed $\mu$. 

\begin{proposition}\label{prop:ren:Nelson} For every $\mu >0$ there is a semi-bounded self-adjoint operator $H^{\rm N}$ so that
	\begin{align}
		\exp\big( -i t H^{\rm N}_\Lambda  \big)
		\exp\big( -  i t   E_\Lambda \big)   
		 \, \xrightarrow{\Lambda \to \infty}  \, 
		\exp\big(-i t H^{\rm N} \big)
	\end{align}
	for all  $t\in \mathbb R$ in the strong sense on $L^2(\mathbb R^6) \otimes \mathcal F$. 
\end{proposition} 
We refer to $H^{\rm N} $ as the \textit{renormalized massless Nelson model}.
Let us note that \cite{Griesemer18} addresses the case of one particle with $\mu=1$. However, the proof can be straightforwardly extended to cover multiple particles and arbitrary values of $\mu$. We sketch the proof in the appendix.

\subsection{Main result and discussion}

Our main result is contained in
the following theorem.
We introduce the effective Schrödinger operator for two particles interacting through the   Coulomb potential
\begin{equation}
	h^{\rm eff}  
	 \, \equiv  
	  \, p_1^2  \, +  \, p_2^2   \, - \,  W(x_1-x_2) \qquad 
	\t{with}
	\qquad 
	W (x)  \, \equiv  \, 
4\pi^2 |x|^{ -1 }
\end{equation}
with domain $D(h^{\rm eff}) = H^2( \R^6)$.

\begin{theorem}\label{thm:main} 
	There exists  $C>0$ and $  \mu_0 \geq 1 $ such  that
	for all  $\varphi \in H^2(\mathbb R^6)$ with $\| \varphi \|_{L^2} = 1 $
	\begin{align}
		\Big\| \Big(e^{-i t H^{\rm N}} -   e^{-it h^{\rm eff}  } \otimes \textnormal{\textbf{1}} \Big) \varphi \otimes \Omega \Big\|^2_{\mathscr H }
		\le
		\frac{ C  \,  ( 1 + |t| \,  ) \ln \mu }{\mu^{1/6}}  	 \,
	\|	 \vp	\|_{H^2} 
	\end{align}
	for all $t\in \mathbb R$ and $\mu  \geq \mu_0  $.
\end{theorem}
\begin{remark}
Giving up the convergence rate, a density argument shows that  as a corollary
	\begin{align}
\label{eq:cor}	
	e^{-i t H^{\rm N}}( \textnormal{\textbf{1}} \otimes P_\Omega ) \xrightarrow{\mu \to \infty} e^{-it h^{\rm eff}} \otimes P_\Omega
	\end{align}
	 in the strong sense in $\mathscr H$, for all $t \in \R .$
	 Here of course  
	 $P_\Omega  = \ket{\Omega}\bra{\Omega}$.
\end{remark}

The theorem demonstrates that in the limit $\mu \to \infty$, the particles effectively decouple from the quantum field, which remains in the initial Fock vacuum state. Instead of coupling to the field, the particles evolve under an instantaneous pair interaction described by the attractive Coulomb potential. Physically, the effective potential  arises from the exchange of virtual excitations,
while the emission or absorption of real bosons is energetically suppressed.   We remark that it is a particular feature of the scalar field model that two particles of same charge experience an attractive Coulomb interaction, rather than the repulsive one typically expected in QED.  

Let us now compare
our results with the literature.  
To this end, we denote by
$\omega (k) \equiv (k^2 + m^2)^{1/2}$
the dispersion relation of a field of mass $m \geq 0$. 
As mentioned earlier, the limit   $\mu \to \infty$ was first studied by Davies \cite{Davies79}, who proved   \eqref{eq:cor} for the UV-regularized Nelson model \eqref{eq:HN}  with
\textit{fixed} cutoff $\Lambda<\infty$.
In particular, since  energy renormalization was not considered, 
the 
   effective potential still depended
on   $\Lambda$.  Hiroshima then revisited this problem in \cite{Hiroshima98}. 
He analyzed the $\mu \rightarrow \infty $ limit of the 
 resolvent
 of  the  massive
 Nelson model,   
 while simultaneously removing the cutoff: 
\begin{equation}
	\label{alpha}  
	 	\Lambda  =  \mu^\alpha
	 	\qquad 
	 	\t{as}
	 	\qquad \mu \rightarrow \infty \ , \qquad 0 < \alpha < 1/4  \ .  
\end{equation}
 In   follow-up work \cite{Hiroshima99},
 Hiroshima analyzed  the emergence of the effective Hamiltonian for the semigroup
$ \{   \exp( -t (H_\Lambda^{\rm N} - E_\Lambda) )   \}_{t>0}$ for $\Lambda = \mu^\alpha$ as $\mu\to \infty$, 
and extended his results
 to the massless case  $m =0$.
More recently, Gubinelli, Hiroshima and L\"orinczi \cite{GubinelliHJ2014} proved the convergence of the semigroup for the renormalized Nelson Hamiltonian with massive Bose field, with the effective Hamiltonian containing a Yukawa potential instead of the Coulomb potential. It is worth noting that the focus on the semigroup stems from the functional integration techniques used in \cite{Hiroshima99, GubinelliHJ2014}. While this method offers its own advantages and insights, it is not suited for studying the unitary group. The generalization to the Nelson model with semi-relativistic particles, that is, with $p^2$ replaced by $(p^2 + 1)^{1/2}$, is addressed in \cite{Takaesu2}, which examines the model with a UV cutoff $\Lambda < \infty$, and in its sequel \cite{Takaesu3}, which considers the simultaneous limit $\Lambda = \mu^\alpha $  discussed earlier. 

\begin{table}[t!]
\centering
\renewcommand{\arraystretch}{1.5} 
\begin{tabular}{|c|c|c|c|l|}
\hline 
\textit{Particles} & \textit{Boson mass}  & \textit{UV cutoff} & \textit{Convergence} &   \\ \hline 
$p^2$  & $m \ge  0$     & $\Lambda <\infty$     &  unitary group   &  Davies \cite{Davies79}  \\  \hline
$p^2$    &   $m>0$ & $\Lambda = \mu^\alpha$     & unitary group     &  Hiroshima \cite{Hiroshima98} \\ \hline
$p^2$      & $m=0$     & $\Lambda = \mu^\alpha$     & semigroup     &   Hiroshima  \cite{Hiroshima99}  \\  \hline
$\sqrt{p^2+1}$     & $m \ge  0 $     & $\Lambda < \infty$    & unitary group     &    Takaesu   \cite{Takaesu2}  \\ \hline
$p^2$     & $m>0$     & $\Lambda = \infty $         & semigroup     &     Gubinelli et al. \cite{GubinelliHJ2014}  \\  \hline
$\sqrt{p^2+1}$      & $m \ge 0$      &  $\Lambda =  \mu^\alpha$     & unitary group     &  Takaesu  \cite{Takaesu3}\\  \hline
$p^2$     & $m = 0$     & $\Lambda = \infty $     & unitary group     &    \textit{present work} \\  \hline 
\end{tabular}
\caption{Previous works on the Nelson model in the weak coupling limit (chronological order).}
\label{tab:example}
\end{table}

Compared to these works, our analysis provides the first derivation of the Coulomb potential starting from the \textit{renormalized massless} Nelson model. 
Moreover --for $H^2(\R^6)$ initial data-- our method establishes an explicit rate of convergence, which was not obtained with the methods previously used.

Before proceeding with the preparations for the proof, we briefly highlight some other related works. 
 In \cite{CarMit24} the authors of the present work
  recently studied the UV-regularized Nelson model \eqref{eq:HN}  in the limit $\mu \to \infty$, 
  demonstrating that the coupling to the quantum field introduces radiative corrections to
   the particles'  
  free dispersion relation; 
  these 
  become significant on timescales 
   longer than those originally considered by Davies \cite{Davies79}, 
  or for initial states with large momentum. However, that analysis focused on the single-particle case, where no effective interaction arises. A different regime, where the interaction term in \eqref{eq:HN} is multiplied by an additional factor of $\mu^{1/2}$, is studied in \cite{Stefan2,TenutaTeufel} for the Nelson model with UV-regularization. Using a space-adiabatic theorem, the authors derive an effective two-body interaction between the particles. 
  The key difference, and a notable challenge compared to the weak coupling limit, is that the particles are dressed with a non-trivial rotated vacuum state.  
  On a similar direction, we note extensions of the weak coupling limit to simplified QED models \cite{Arai,Takaesu,Hiroshima3} and investigations into the effect of enhanced binding \cite{HiroshimaSasaki,HiroshimaSasaki2}.
  Finally, let us mention that for cold Fermi gases, 
  similar effective interactions for impurity particles 
  can be mediated by excitations of the Fermi ball,   playing the role of the scalar Bose field. See e.g. \cite{Pickl,Jeblick2}.

\vspace{1cm}

\textit{{{Outline of the proof}.}}
The first step in  our analysis
is a unitary representation of $H^{\rm N }$
in terms of a  {dressed} Hamiltonian
$H_K$, where $ 1 \ll K  \ll \mu $ 
is a large parameter. 
This operator
contains an  effective 
UV cutoff $K $ in \eqref{eq:HN}, 
at the expense of introducing
several new  terms.
In order to simplify the analysis, 
we introduce yet another Hamiltonian, denoted by $H_{\ve,  K}$, 
which results from elimination of lower order terms in $H_K$.
Up to small corrections,
 it suffices to compare the dynamics generated by  $ H \equiv H_{\ve, K }$  with 
  the one  generated by the effective Hamiltonian $h^{\rm eff}$. 
   This comparison is based on an expansion that exploits fast oscillations of the unitary operator $e^{-i t H }$. More concretely, after performing a first-order Duhamel expansion --see Eq. \eqref{exp}-- the contribution  $\mu T$ in $H $ induces fast oscillations of  $e^{-i t H }$ as $\mu \rightarrow \infty$. 
   			Using a suitable integration-by-parts formula --see Eqs. \eqref{eq1} and \eqref{A:exp}-- we apply a stationary-phase-type argument to expand the difference between the dynamics. In this expansion, two types of terms emerge: \textit{Subleading terms}, which vanish due to cancellations induced by the oscillations, combined with suitable operator and energy estimates; and  \textit{leading-order terms}, which arise from states in the subspace
   			 $\t{ker}(T) = \{\phi \otimes \Omega : \phi \in L^2(\mathbb R^6)\}$,
   			  where oscillations are absent. Physically, these terms correspond to the self-energy and the effective interaction between the particles. The leading-order terms involve zero final bosons and are commonly referred to as virtual exchange processes. In particular, we prove in Lemma \ref{lemma:cancel} that the Coulomb interaction $W( x_1 - x_2)$
emerges as an effective interaction from these processes. 
 
\smallskip 

\textit{Organization of the article.} In Section \ref{sec:dressed}, we introduce the dressed Hamiltonian $H_K$ and its relation to $H^{\rm N}$. This is followed by a section on preliminary operator estimates, including Lemma \ref{lem:energy:bounds}, which is crucial for controlling the regularity of time-evolved states (see Lemma \ref{lemma:E}) necessary for exploiting the oscillations induced by the unitary evolution. The main part of the proof is presented in Section \ref{sec:main:proof}, beginning with an approximation of $H_K$ in terms of a simplified dressed Hamiltonian $H_{\varepsilon , K}$, where subleading terms are removed. The comparison of the dynamics generated by $H_{\varepsilon ,K}$ and $h^{\rm eff}$ is made in Section \ref{subsection:main} and comprises most of the analysis in this article. The main cancellations between the two generators are established in Lemmas \ref{lemma:cancel} and \ref{prop3}. In the appendix, we provide some details on the proofs of Proposition \ref{prop:ren:Nelson} and Lemma \ref{lemma:dressed:Nelson}.

  \subsection{Notation}
Throughout this paper we implement the following notation.

\noindent - We drop
the subscript $\mathscr H $ in the  Hilbert space norm $\|  \Psi \|_{\mathscr H } = \| \Psi\|$. 
Often we also  identify $\mathscr H  =  L^2( \R^6 , \mathcal F)$
and write for the inner product
 \begin{equation}
 	\label{fock}
 	 \< \Phi, \Psi \> 
 	  = \int_{\R^6} \<  \Phi(X), \Psi(X)\>_\F  \, d X \  , \qquad \Phi , \Psi \in \mathscr  H \ . 
 \end{equation}

\noindent  - We consider the space $L^\infty L^2$
  of functions 
  $f : \R^6 \times \R^3 \rightarrow \mathbb C  $
  with  finite norm 
  \begin{equation}
  	 \|   f  \|  \equiv  
  	   \|   f \|_{L^\infty L^2}
  	    = \sup_{X\in \R^6} 
  	   \bigg(	    \int_{\R^3} | f (X, k)|^2 dk		\bigg)^{1 /2 } \ . 
  \end{equation}
If 
 $T_f  :L^2(\R^6 ) \rightarrow L^2(\R^6)  \otimes L^2(\R^3)$
 is the bounded operator 
   $ ( T_f \psi)  (X, k) \equiv  f(X,k)\psi ( X )$, 
   then 
  the operator norm of $T_f$ and $\| f \|$   coincide.

\noindent -
For $F \in L^2 (\R^3, dk)$
 we denote by $a ( F)$ and $a^* (F )$
 the standard creation and annihilation operators on Fock space
 $ \F $. In particular 
  \begin{equation}
 	a ( F ) = \int_{\R^3} \overline{F (k)}  a_k  dk \ ,  \qquad
 	a^* ( F ) = \int_{\R^3}  {F (k)}  a^*_k  dk  \ . 
 \end{equation}
We also define the Fock-Segal operator and the bosonic number operator 
\begin{equation}
	 	\label{phi}
	 \phi ( F ) \equiv a (F ) + a^* ( F )   \ , 
	 \qquad N \equiv \int_{\R^3 } a_k^* a_k d  k \ . 
\end{equation}
 
\noindent -
For  $f \in L^{\infty}L^2 $ 
we 
extend
the operators to $\mathscr  H $ via 
$   [  a(f) \Psi ]  (X) \equiv a (f_X) \Psi(X )$ 
with $f_X(k)   =  f(X,k)$. 
Similarly for $a^* (f)$ and $\phi (f)$.

\section{The dressed Hamiltonian \label{sec:dressed}}
Here and in the sequel, 
we introduce the   following 
form factors
in the interaction term of  \eqref{eq:HN} 
\begin{align}
	 G_1 (k )    
	 \equiv  
	 | k |^{ -1 /2 }
	 e^{-ik x_1} \otimes \textbf{1}
	 \qquad 
\t{and}
	 \qquad 
	  G_2 (k )    
	 \equiv  
	 | k |^{ -1 /2 }
	\textbf{1} \otimes  e^{-ik x_2} 
\end{align}
regarded as operators on $L^2(\R^6)$. 
In addition, for $\Lambda>0$
we introduce their cutoff versions  
\begin{equation}
	G_{\Lambda, 1 }(k)  
	\equiv  
G_1(k) \1_{| k | \leq \Lambda}
	\qquad
	\t{and}
	\qquad 
	G_{\Lambda, 2  }(k)  
\equiv  
G_2 (k) \1_{| k | \leq \Lambda}  \ .
\end{equation}
From now on,  we   no longer use  the tensor product  $\otimes$ in our notation. 
We also  identify the functions $G$ with elements of $L^\infty L^2. $

Our analysis is based on a representation of the Hamiltonian $H^{\rm N}$ via the following unitary dressing transformation (also called Gross transformation)
for $ K >  0 $ 
\begin{align}\label{eq:Gross:trafo}
U_{K} 
\equiv 
 \exp\Big(  \mu^{1/2} 
 \sum_{i=1 ,2 }
(  a^*(B_{K,i}) - a( B_{K,i})) \Big)    \quad \text{with}\quad 
B_{K, i } (k) \equiv   \frac{G_i (k)  }{k^2 + \mu 
	| k |} \1_{|k|\ge K}  
\end{align}
for $ i =1,2$. Observe that  $K$ acts as an upper cutoff in  $G_{K , i}$, 
but as a lower cutoff in $B_{K,i}$. 
The next statement is adapted   from \cite{Griesemer18}, 
and a sketch of the proof is provided in the appendix.

\begin{lemma}\label{lemma:dressed:Nelson}
For all $K >  0$
	 we have the operator identity $H^{\rm N} = (U_K)^* H_K U_K$ with dressed Hamiltonian
\begin{align}\label{eq:H:ren}
H_{K} & = p_1^2 + p_2^2 + \mu T  + \mu^{1/2} \sum_{i=1,2}
 \Big( \phi(G_{K,i} ) + 2 a^*( k B_{K,i} ) \cdot p_i +  2  p_i \cdot a ( k  B_{K,i} ) \Big) \notag\\
& \quad + \mu    \sum_{i=1,2} \Big( 2 a^* ( k B_{K,i}  ) a ( k  B_{K,i}  ) +  a^* ( k B_{K,i}  ) a^* ( k B_{K,i}  ) +  a ( k B_{K,i}  ) a ( k B_{K,i}  ) \Big)  \notag\\[2mm]
& \quad +         E_K     +    V_K(x_1-x_2) ,
\end{align}
with  quadratic form domain $Q(H_K) = Q(p_1^2 + p_2^2 +   \mu T )$.
Here, $E_K$ is the energy defined in \eqref{E}, 
and we introduced the potential  
\begin{equation}
	\label{VK}
	V_K(x ) 
	    \,  =  \, 
(-2)  \mu \, 
\int_{   |k |\geq K 	}
\frac{\cos(k  x ) ( \mu  | k | +2 k^2 ) \, d k }{ |k|
	 ( \mu  | k |+k^2)^2} .
\end{equation}
\end{lemma}

\begin{remark}  
	Relative to \cite{Griesemer18}, 
	here we obtain a new term corresponding to $V_K$ due to effective interactions between the two particles. 
An integral  estimate shows that  there  is  $C>0$ so that 
\begin{align}
	\label{VK:L^2}
	\|    V_K\|_{L^2}^2 
	 = 4 \mu^2 
	 \int_{ | k |	 \geq K 	}	
	  \bigg|   \frac{ (   \mu |k| + 2k^2 )  }{ |k | ( \mu |k| + k^2)^2 }				\bigg|^2 dk 
	    \leq  \frac{C}{K} \ , 
\end{align}
for all $ \mu, K >0$. 
\end{remark}

\section{Preliminary estimates}

  First, we  introduce
some  auxiliary
 operators. 
For $0\le a<b \le \infty $, we set
 \begin{align}\label{eq:a:b:notation}
 	G_{a,b , i } =   \1 (a \le |k| \le b) \, G_i   \    , 
 	\quad N_{a,b} 
 	=
 	 \int\limits_{a \le |k| \le b } a^*_k a_k d k  \ , \quad 
 	T_{ a,b } 
 	=
 	 \int\limits_{  a \le |k| \le b } |k| a^*_k a_k  d k .
 \end{align}
 In particular $N_{0,\infty} =N$ and $T_{0,\infty} = T$. 
 Similarly, 
 we introduce the projections 
 \begin{equation}
 	P_{a,b} = \1  (N_{a,b}  =0) \ , \qquad Q_{a,b}= \1 (N_{a,b} \ge 1) \  . 
 \end{equation}
Note that $P_{0,\infty} =  |\Omega \rangle \langle \Omega|$ coincides with the projection on the   vacuum. 
  Since $T_{a,b} \ge a  N_{a,b}$ and since $N_{a,b}^m$ has a bounded inverse on $\text{Ran}(Q_{a,b})$, we have
 \begin{align}
 	\| ( T_{a,b} )^{-m} Q_{a,b}  \| \le  a^{-m} \| ( N_{a,b} )^{-m} Q_{a,b}  \| \le  a^{-m}.
 \end{align}
 We will further need the following resolvent type operator
 \begin{align}
 	\label{def:resolv}
 	R_{a , b } =  Q_{a,b} (T_{a,b} + E_K + (p_1^2+p_2^2) /\mu )^{-1} Q_{a,b}
 \end{align}
 where we recall that $E_K > 0$, hence this is well defined for all $a\ge 0$. 
 Moreover, for $a>0$ it satisfies $\| R_{a,b}^m \| \le a^{-m}$ for any $m > 0$, and
 for any $ i =1,2 $
 \begin{align}
 	\| T_{a,b} ^{1/2} R_{a,b}^{1/2} \| \le 1 , \qquad \| ( T_{a,b} + 1) ^{1/2} R_{a,b}^{1/2} \| \le C ( 1 + a^{-1/2}) ,\qquad \| p_i R_{a,b}^{1/2} \| \le \sqrt \mu.
 \end{align}

\subsection{Integral estimates}

We start by summarizing some relevant norms that will come up many times.
The proof  is     a simple 
calculation using the change of variables $\ell \equiv  \mu^{-1} k$ which is left to the reader. 
 
\begin{lemma}\label{lem:form:factors} 
Let $G_K$ and $B_K$ denote  $G_{K,i}$ and $B_{K,i}$ for any $i =1,2$. Denote $\omega(k) = |k|$. 
\begin{enumerate}  
	\item[\rm (1)] 
	There is $C> 0$ such that for all $K  >  0  $ 
\begin{equation}
	\label{GK}
	 \|   G_K \|  \leq C K  \  , 
	 \qquad 
\| \omega^{- 	 \frac{1}{2}		} G_K \|   \leq C  K^{1/2} \ . 
\end{equation}

\item[\rm (2)] 
For $0 < s \leq  \tfrac{1}{2} $ there is  $C>0$ such that for all
 $\mu ,   K >  0 $ 
\begin{align}
	\label{BK1}
	\| B_K \|    \leq    
	\frac{C}{\mu}
	   \big(   1 + | \ln ( K  / \mu  )| ^{	 \frac{1}{2} } 	  \big)   \ , 
\qquad 
	 \| \omega^{-s}   k B_K  \|  
	  \leq   \frac{C}{ (\mu +  K)^{s} } \ . 
\end{align}
\end{enumerate}
\end{lemma}
\noindent

 \begin{remark}
 In this article we are mostly concerned
 with the case  $ \mu \gg  K \gg 1   $. 
 Hence, we keep in mind the  simplified bounds
 \begin{equation}
 	\label{BK}
 		\| B_K \|    \leq    
 C \ln (\mu / K )^{1/2} \mu^{-1 }  \ , 
 	\quad 
 	\| \omega^{- \frac{1}{4}} k B_K \|  
 	\leq  
 C \mu^{ - \frac{1}{4} }  \ , 
 	\quad 
 	\| \omega^{ - \frac{1}{2} } 
 	k B_K \| 
 	\leq 
 C \mu^{ -  \frac{1}{2} }  \ .
 \end{equation}	 
 \end{remark}

In order to control the large $K$ limit
associated to the effective potentials, we use the following
result. 
Effectively, we want to compare for two states $\psi,\varphi \in L^2(\mathbb R^6)$ the expression
contained in the following lemma for 
 $v (x) = (2\pi)^{-3/2}	\textstyle   \int_{\R^3 } \hat v (k) e^{ikx} d k  $ 
where $\hat v(k)$ is supported only on $ | k | \geq K $. 
We establish the estimate in a general form
for further use. 

\begin{lemma}
	\label{prop:w:estimate}
Let $  v \in L^2 (\R^3)$. 
Then, 
there exists $C >0 $
such that for all $\phi  , \psi \in H^1 (\R^6)$. 
	\begin{equation}
		\Big| 
		\Big\langle \psi , v(x_1- x_2  ) \phi\Big \rangle 
		\Big| 
		\leq 
		C  			  \| v \|_{L^2}   
		\Big( \| \psi \|_{H^1}^2 + \|\phi\|_{H^1}^2   \Big) \ . 
	\end{equation}
Additionally, 
there exists $C >0 $
such that for all $\phi   \in H^2 (\R^6 )$. 
\begin{align}
		\|   v (x_1 - x_2) \phi \|_{L^2} 
		\, 	 \leq  \, 
		C \| v\|_{L^2} 
		\min
		\big\{ \,  	\|  (  1+p_1^2) \phi \|_{L^2} \,    , 	 \, \|  (  1+p_2^2) \phi \|_{L^2}  		 \, 	\big\} \ . 
\end{align}
\end{lemma}

\begin{proof}
	Let $D \equiv \Big\langle \psi , \int dk \widehat v (k) e^{-ik(x-y)} \phi\Big \rangle . $ 	First, shift $x\mapsto x+y$ and define the shifted function $\psi_y(x)  = \psi(x+y,y)$, and then use Cauchy--Schwarz and Parseval to find 
	\begin{align}
		| D | \le \| \widehat v \|_{L^2_k} \ \Big\|  
		\int \big\langle \psi_y (x)  , \phi_y (x)  \big\rangle_{L_y^2} e^{ikx} dx \Big\|_{L^2_k} =
		 \| v \|_{L^2} 
		  \Big\|  \big\langle \psi_y (x)  , \phi_y (x)  \big\rangle_{L_y^2} \Big\|_{L^2_x}.
	\end{align}
	The triangle inequality, Cauchy--Schwarz inequality, 
	and Young's inequality  imply  
	\begin{align}
		 		  \Big\|  \big\langle \psi_y (x)  , \phi_y (x)  \big\rangle_{L_y^2} \Big\|_{L^2_x } 
		 		  	\leq 
		 		  	 \int d y 
		 		  	 \|	 	 {  \psi_y} 	 	 \phi_y	\|_{L_x^2} 
		 		  	 \leq 
		 		  	 		 		 \int d y 
		 		  	 		 		\|  \psi_y\|_{L^4_x}
		 		  	 		 				 		  	 		 		\|  \phi_y  \|_{L^4_x}
		 		  	 		 				 		  	 		 		\leq 
		 		  	 		 				 		  	 		 		\int dy 
\big(	\|  \psi_y\|_{L^4_x}^2 +   
		\|  \phi_y\|_{L^4_x}^2 	\big)  \ . 
	\end{align}
On the other hand, 
  the Sobolev embedding $H^1(\R^3) \rightarrow L^6(\R^3)$  implies 
  for any $\psi \in L^2(\R^3)$
\begin{equation}
	 \|	 \psi	\|_{L^4}^2
	 \leq 
	 \| \psi \|_{L^2}^{1/2}  \|  \psi \|_{L^6}^{3/2}
	 \leq  C 
	 \| \psi \|_{L^2}^{1/2}  \|  \psi \|_{H^1}^{3/2}  \ . 
\end{equation}
 We put the  last three inequalities together and use 
  $ \| \psi\|_{L^2}\leq \| \psi\|_{H^1}$  to finish the proof. 
	
	For the second part, we do the same change of variables and
	use the Sobolev embedding $H^2(\R^3)  \rightarrow L^\infty(\R^3)$ 
	to obtain 
	\begin{align}
\notag
		 \iint dx dy | v (x -y )|^2 | \phi (x,y)|^2 
 & 		  = 
		  		 \iint dx dy | v (x   )|^2 | \phi (x+ y ,y)|^2  \\
		  		 \notag
	& 	  		 \leq 
		  		 \|  v\|_{L^2}^2 
		  		 \int d y   \sup_x | \phi(x,y)|^2  \\ 
		&   		 \leq 
		  		 C 
		  		 		  		 \|  v\|_{L^2}^2    \int d y 
		  		 		  		 \| \phi ( \cdot, y) \|_{H_x^2}^2 
		  		 		  		  =  
		  		 		  		 		  		 C 
		  		 		  		 \|  v\|_{L^2}^2    
		  		 		  		 \|  \phi  \|_{L^2_y H_x^2}
	\end{align} 
Interchanging the roles of $x$ and $y$, this gives the desired estimate. 
\end{proof}

\subsection{Operator estimates}

First, we state
useful bounds for creation and annihilation operators.

\begin{lemma}[Operator bounds] \label{lem:operator:bounds} 
Let $f, g  \in L^\infty L^2$ 
be supported  
 on  $ \R^6 \times   \O \subset \R^9$.
Let us denote  
$N_\O  \equiv \int_{\O  } a_k^*a_k dk $
and
$T_\O  \equiv \int_{\O  } \omega(k)  a_k^*a_k dk $ for some measurable function  $\omega : \R^d \rightarrow [ 0 , \infty) $. 
Then, it holds that
		\begin{align}
			\label{af1}
		\|   a(f) \Psi \| 		 
		& \leq	
		\| f\| 	 \| N_\O ^{  1/2 }  \Psi \| 
		    & &  \forall \Psi  \in D (N_\O ^{1/2})			 \ , 
			\\
		\label{af2}
		\|   a(f) \Psi \| 		 
		&  \leq	
		\|  \omega^{ - \frac{1}{2}} f\| 	 \| T_\O ^{  1/2 }\Psi \|  
	 & & 	 \forall \Psi  \in D (T_\O ^{1/2})			 \ . 
	\end{align}
Additionally,  
\begin{align}
			\label{af3}
	 		\|  (1+N_\O )^{- \frac{1}{2}} a(g)  a(f) \Psi \| 		 
 &   \leq	
	 \|  	 \omega^{ - \frac{1}{4}} f\|
	 \|  \omega^{ - \frac{1}{4}} g\| 
	 \|   T_\O ^{  1/ 2 } \Psi \|    
	 & & 
	  \forall \Psi  \in D (T_\O ^{1/2})	 \ . 
\end{align}

\end{lemma}

\begin{proof}
 	The first two inequalities follow from a simple application of Cauchy--Schwarz.
The third one is taken from \cite{Griesemer18}, 
which originates in  a paper of Nelson \cite{Nelson64}. 
Let us give here the  proof for completeness. 
First,      consider $f,g \in L^2(\R^3,  \omega^{-1/2} dk)$ and $\Psi\in D(T_\O^{1/2})$ and compute 
\begin{align}
	\notag 
	\|	 a (f ) a(g) \Psi 	\|_\F 
 & 	 \ \leq 	 \ 
	\int_{ \O \times \O  } 
	   | f( k_2) |   | g (k_1)  |    \| a_{k_1 } a_{k_2} \Psi		\|_{\F } 	\ d k_1  d k_2   \\ 
	  	\notag 
& 	 \  \leq  \ 
	 \|   \omega^{- \frac{1}{4}} f		\|_{L^2}
	 	 \|   \omega^{- \frac{1}{4}} g		\|_{L^2}
	 	 \Big( 	\int_{ \O \times \O  }  
	 	  \omega(k_1)^{1/2 }
	 	  \omega(k_2)^{1/2}
	 	 \|  a_{k_1} a_{k_2} \Psi \|^2_\F			 \ d k_1   d k_2  	\Big)^\frac{1}{2} \\
	 	 & 	 \  \leq  \ 
	 	 \|   \omega^{- \frac{1}{4}} f		\|_{L^2}
	 	 \|   \omega^{- \frac{1}{4}} g		\|_{L^2}
	 	 \Big(		\int_{ \O \times \O  }  
	 	 \omega(k_1)
	 	 \|  a_{k_1} a_{k_2} \Psi \|^2_\F 			 \ d k_1   d k_2  	\Big)^\frac{1}{2}   \  , 
\end{align}
where in the last line we used 
$\omega_1^{ 1/2 } \omega_2^{ 1/2 }
 \leq \frac{1}{2} (\omega _1  + \omega_2)$, 
a change of variables $k_2 \mapsto k_1$, 
and the 
canonical commutation relations (CCR) for the field operators. 
Next,  we use the CCR to calculate
$
	a_{k_2}^* a_{k_1}^* a_{k_1} a_{k_2} 
	 = 
	 \delta(k_1 - k_2 ) a_{k_1}^* a_{k_1}
	 +   a_{k_1}^* a_{k_1} a_{k_2}^* a_{k_2} $
and integrate over $k_1 $ and $k_2$, to find 
\begin{align}
	\|	 a (f ) a(g) \Psi 	\|_\F 
	& 	 \  \leq  \ 
	\|   \omega^{- \frac{1}{4}} f		\|_{L^2}
	\|   \omega^{- \frac{1}{4}} g		\|_{L^2} 
\, 
 \< \Psi,  T_\O ( 1+ N_\O   ) \Psi\>_\F^\frac{1}{2 }   \ . 
\end{align}
The previous bound can then be extended to 
$f,g \in L^\infty L^2$ and $\Psi\in \mathscr H$
using \eqref{fock}. 
This is the desired estimate after changing $\Psi \mapsto(N_\O+1)^{-1/2}\Psi. $
\end{proof}

The next lemma is motivated from \cite[Lemma 3.1]{Griesemer18}. Due to the different choice of scaling, we provide some details.

\begin{lemma} 
	\label{lem:energy:bounds} 
	There exist $ C , K_0 > 0$ such that 
	for all 	$ \mu>0$ and $ K \geq K_0 $
	\begin{equation}
\frac{1}{2}  \, (p_1^2 + p_2^2 + \mu T) - C K 
			  \   \leq 		 \ 
			  H_K 
 \			   \leq  \ 
			  \frac{3}{2} \,  ( p_1^2 + p_2^2 + \mu T )  + C  K 
	\end{equation}
in the sense of quadratic forms on $Q(H_K) = Q(p_1^2 + p_2^2 + \mu T)$.
\end{lemma}

\begin{proof} 
	Let $G_K$,  $B_K$  and $p$
denote  $G_{K,i}$,  $B_{K,i}$ and $p_i$ for any $i =1,2$. 
First, 	we claim that we can prove the following bounds, 
	for all $\mu , K > 0 $
	\begin{align}
		\label{op1}
		|   \< 	  \Psi, a ( G_K) \Phi			\> | 
		& \,   \leq 		 \, 
		C K^\frac{1}{2 } \,  \|	 \Psi	\|
		\|	T^\frac{1}{2} \Phi	\| 		\\
		\label{op2}
		|   \langle   \Psi , p \cdot  \, a(k B_K) \Phi \rangle | 				
		& 		\,		 \leq		\,		
		C K^{- \frac{1}{2}}
		\| p \Psi  \|	\|	T^\frac{1}{2} \Phi	\| 	 	\\ 
		\label{op3}
		|    \langle \Psi ,  a^*(k B_K) a(k B_K) \Phi \rangle | 	
		& \, \le	\, 		C K^{-  1 }
					 \| T^{1/2} \Psi \|			 \| T^{1/2} \Phi \| \\
		\label{op4}
		|	 \langle 	\Psi,  a(k B_K) a (k B_K ) \Phi \rangle		|
		& 	 \, \leq  \, 
		C   	(\mu  + K)^{- \frac{1}{2}}  
		\|	 (1  + K^{ - \frac{1}{2}} T^\frac{1}{2}) \Psi	\| 
		\|   T^{\frac{1}{2}}\Phi 		\|  \\
		\label{op5}
		|	  \<	 \Psi, V_K(x_1 - x_2 ) \Phi		 \>		|
		&   \,   \leq  \, 
		C K^{  - \frac{1}{2}}  \, 
		\|   (1 + p^2 )^\frac{1}{2} \Psi		  \| \| (1 + p^2 )^\frac{1}{2} \Phi  \|
	\end{align}
	and all $\Phi, \Psi \in  Q(p_1^2+p_2^2 + \mu T )$.
Additionally,  we note that  
  $   E_K \leq C  ( 1 +  K)  $ for    $ K  >0 $. 
The estimates 
	 \eqref{op1}--\eqref{op5}
	combined with Young's inequality
	$a b \leq  \frac{\epsilon}{2} a^2 +   \frac{1}{2 \epsilon } b^2$
	implies the  desired energy estimates for $H_K$,
	provided we take $ K \geq K_0$ large enough. 
	Note we can choose $K_0$ independently of $\mu$.

	The bounds \eqref{op1}--\eqref{op3} follow directly from
	the integral and operator bounds in 
	\eqref{GK}, \eqref{BK1}  and \eqref{af2}.
For the fourth  bound,
observe that $k B_K$ is supported  on
$| k | \geq  K $. Thus,  we obtain
thanks to the Cauchy--Schwarz inequality 
and Eqs. \eqref{BK1} and \eqref{af3}
\begin{align}
	\notag
		 	|	 \langle 	\Psi,  a(k B_K) a (k B_K ) \Phi \rangle		|
		  & 	\leq 
		 	\|   (1 + N_{ K , \infty })^{ \frac{1}{2 }} \Psi 			\| 
		 	\|	 (1 + N_{ K , \infty } )^{ - \frac{1}{2}} a(k B_K) a (k B_K ) \Phi 	\| \\
		 		\notag
		 	& \leq 
		\|   (1 + N_{ K ,\infty })^{ \frac{1}{2 }} \Psi 			\| 
		 	\| |k|^{-\frac{1}{4}} k B_K  \|^2 \| T_{K,\infty}^\frac{1}{2} \Phi  \| \\
\label{aa}
		 	& \leq 
		 	  C 	 (\mu  + K)^{- \frac{1}{2}}  
		 	  		 		\|   (1 + N_{K , \infty })^{ \frac{1}{2 }} \Psi 			\| 
		 	  		 	\|   T_{K,\infty}^{\frac{1}{2}}\Phi 		\| \ . 
\end{align}
It suffices to use   
$N_{ K , \infty} \leq K^{-1} T_{K,\infty}  $
and
$T_{ K , \infty} \leq T$. 
Next, we control the contributions from  
  $V_K(x_1 - x_2)$
using  $\| V_K \|_{L^2 }  \leq C K^{-1/2}$ 
and Lemma \ref{prop:w:estimate}. This gives \eqref{op5} 
and  finishes the proof. 
\end{proof}

Let us observe that from the proof we
learn that the more intricate term is given by 
$a (k B_K)a(k B_K)$. 
In the following lemma, we record  estimates for this term that will be relevant for our  analysis.   
Recall that $R_{a,b}$ was introduced in \eqref{def:resolv}. 

\begin{lemma}\label{lemma:A}
	Let $B_K$ and $ p $ denote $B_{K, i } $ and $ p_i $ for any $ i  =1 , 2$. 
	Then, 
	there exists $C>0$
	such that for all $\mu,  K > 0$
	and
 $\Psi \in Q(p_1^2+p_2^2 + \mu T ) $.
 \begin{align}
 	\label{A10}
 	\|	 ( 1 + N )^{ - \frac{1}{2}}	 a ( k B_K) a (k B_K )   \Psi 		\|
 	  & 	\leq 
 	C \mu^{- \frac{1}{2}}  \|	 T^{1/2} \Psi 		\|  	 
 \end{align}
Additionally,  we have the operator bounds
 	\begin{align}
 		\label{A20}
\|	 (1 + N)^{ - \frac{1}{2}} 		a ( k B_K) a (k B_K )  R_{ K , \infty} 	\| 
     			& \leq 			
     			  C K^{   - 1 /2 } \mu^{ - 1/2 }
     			 	\\
     			 	\label{A30}
     			\|	R_{ K , \infty}  	a^* ( k B_K) a^* (k B_K )   (1 + N)^{ - \frac{1}{2}} 	 	\| 
     			& \leq 		
     			     			  C K^{   - 1 /2 } \mu^{ - 1/2 } 
 \end{align}
 \end{lemma}
\begin{proof}
From  \eqref{aa}
we obtain 
	\begin{equation}
				\label{op6}
		|	 \langle 	\Psi,  a(k B_K) a (k B_K ) \Phi \rangle		|
		   	 \, \leq  \, 
		C   		  \mu^{- \frac{1}{2}}
		\|	 (1 +   N_{K,\infty } )^{ \frac{1}{2} }	 \Psi	\| 
		\|   T_{K,\infty}^{\frac{1}{2}}\Phi 		\|   \ . \\
	\end{equation}
	The   bounds in the lemma
	then follow from \eqref{op6}, 
	$\| T_{K,\infty}^{1/2} R_{K, \infty}^{1/2}\| \leq 1 $ and 
	 $\| R_{K,\infty}^{1/2}\|\le C K^{-1/2}$. 
\end{proof}

\section{Proof of the main result \label{sec:main:proof}}
Throughout this section, 
we let $\vp \in H^2(\R^6)$
be such that $\| \vp \|_{L^2}=1$.
In particular, all constants will be  independent of  $\vp$. 
Let us record here the proof of the main theorem.
The proof relies in the use 
of an \textit{intermediate dynamics}, 
  generated by the simplified Hamiltonian
for $ 0 < \ve \le  K < \infty$  
\begin{equation}
	H_{\ve , K } 
	\equiv  P^2+ \mu T  +   E_K +  
	V_K+
	\sum_{i=1,2}
	\Big(  \sqrt \mu \phi (G_{\ve, K,i  })
	+ \mu ( A_{K, i } + A_{K,i }^*) \Big)
\end{equation}
where here and in the sequel we denote 
$P \equiv  (p_1, p_2)$ and use the notation 
\begin{equation}
	A_{K,i}  \equiv  a (k B_{K,i }) a (k B_{K,i })   \ .
\end{equation}

\begin{proof}[Proof of Theorem \ref{thm:main}]
	
The first step is to
  replace $H^{\rm N}$ by $H_K$ in the time evolution.
  To this end, 
  observe that  there is a constant $C >0 $ such that for all $\mu \geq 2 K  >0  $
		and
		$\Psi \in L^2(\mathbb R^6 ) \otimes \mathcal F$:
		\begin{align}
			\| (U_K^\# - 1 ) \Psi \|
			& \ 	 \le  \ 
			2 \sqrt \mu  \| B_K \| 
			 \|( N +1)^{1/2} \Psi \|   \ \le \ 
			C \mu^{-1/2} \ln (\mu / K)^{1/2}  \|( N +1)^{1/2} \Psi \| 
		\end{align}
 with $U_K$ defined in \eqref{eq:Gross:trafo} and $\#\in \{ \cdot, * \}$. The inequality is a   consequence of the spectral theorem, 
			the bounds for $\|a (f) \Psi  \|$ in   \eqref{af1},  and the bound for  $\| B_K\| $ in  \eqref{BK}. 
			Thanks to   this bound  and Lemma \ref{lemma:dressed:Nelson}, we find 
	\begin{align}
		\| ( e^{-it H^{\rm N}} - e^{-it h^{\rm eff} } ) (\vp \otimes \Omega ) \| 
		& 	\leq  
		\| ( e^{-it H^{\rm N}} U_K  - U_K e^{-it h} ) (\vp \otimes \Omega )  \| 
		+	
		C \mu^{-1/2} (\ln \mu / K )^{1/2}  \notag \\[2mm]
		& =  \| ( e^{-i t H_K} - e^{-it h^{\rm eff}})  (\vp \otimes \Omega )  \| 
		+		  C \mu^{-1/2} (\ln \mu / K )^{1/2}    \ . 
	\end{align}

For the second step, we make use of the intermediate dynamics $e^{- i t H_{\ve , K}}$. 
Namely, the triangle inequality   implies that  
\begin{equation}
	\|   ( e^{- i t H_K }    - e^{- i t h^{\rm eff}} ) \varphi \otimes \Omega \|
	\leq 
	\|   ( e^{- i t H_K }    -  e^{- i t H_{\ve,K } }  ) \varphi \otimes \Omega \|
	+ 
	\|   ( e^{- i t H_{\ve,K} }    - e^{- i t h^{\rm eff}} ) \varphi \otimes \Omega \| \ .  					
\end{equation}
The two terms on the right side 
are controlled with Lemma \ref{lemma:removal}
and  Proposition  \ref{thm2}, respectively. 
Namely, there exist $C>0$ 
and $K_0>0$
so that for $ \mu   \geq 2  K $  and $ K \geq K_0 $
\begin{equation}
		\|   ( e^{- i t H_K }    - e^{- i t h^{\rm eff}} ) \varphi \otimes \Omega \|^2 
		\leq C
		(1+|t|)
\big(		\| \vp\|_{H^1}^2	+ 	\| \vp\|_{H^2} 	\big) 
	\Big(		
\frac{ K }{\mu^{1/2}} 
\ln( K / \ve ) 
+
\frac{1}{K^{1/2}}
+
\ve 		
+ 
\ve^\frac{1}{2}K^\frac{1}{2}
\Big) \ .
\end{equation}
Finally, we choose  $
	 K  =  \mu^{ 1/3 }$ and $\ve = \mu^{-1}$
	 and take $\mu \geq \mu_0 $ large enough so 
	 that 
	  $ \mu_0 \geq 2 \mu_0^{1/3}$ 
	  and 
	 $ \mu_0^{1/3} \geq K_0$ are satisfied. 
	 We use $ \|\vp  \|_{H^1}^2 \leq 2 \| \vp \|_{L^2} \| \vp \|_{H^2}$
	  to finish the proof. 
\end{proof}

 The rest of this section is organized as follows.
 In Subsection 	  \ref{subsection:apriori} 
 we record 	 a
  collection of estimates that  we repeatedly use in the proofs of Lemma \ref{lemma:removal}
  and Proposition \ref{thm2}. 
  In Subsection  \ref{subsection:removal} we prove Lemma \ref{lemma:removal}, 
  and in Subsection  \ref{subsection:main} we prove Proposition \ref{thm2}
  which contains the main expansion leading to the proof of Theorem \ref{thm:main}.

\subsection{Apriori estimates}
\label{subsection:apriori}

In this subsection, we establish two classes of estimates. 
The first one we refer to as 
 \textit{energy estimates}. They follow from  Lemma  \ref{lem:energy:bounds}
and the analogous bounds for the simplified Hamiltonian.  The second one 
refers to  applications 
of the operator bounds of  Lemma   \ref{lem:operator:bounds}
and  	 \ref{lemma:A}  
to time evolved states.
 
\begin{lemma}\label{lemma:E}
Let $\Psi _t^\flat$ denote either 
$  e^{- i t H_K} \varphi \otimes \Omega $
or 
$e^{- i t H_{\varepsilon , K  }  }    \varphi \otimes \Omega $. 
Then, there are $  C , K_0 > 0$
such that for all
$ \mu \geq K \geq K_0 \ge \ve >0$,  all $t \in \R $
and $ i = 1,2 $
	\begin{align} 
		\label{E1}
		\tag{E1}
		\|    p_i   \Psi_t^\flat \| 			
 & \ 		\leq 			 \ 
		C K^{\frac{1}{2 }} \|  \vp\|_{H^1}  \\
				\label{E2}
		\tag{E2}
		\|   T^{1/2} \Psi_t^\flat \| 			
		 & \ \leq 	 \ 
		C	\mu^{ - \frac{1}{2}}  K^\frac{1}{2} 
		\|  \vp\|_{H^1}  \ . 
	\end{align} 
\end{lemma}

\begin{proof}
The energy estimates for $\Psi^\flat_t = e^{- i t H_K} \vp \otimes \Omega$
are a straightforward consequence of Lemma \ref{lem:energy:bounds}. 
Additionally, 	we 
	observe that the bounds
	\eqref{op1}--\eqref{op5}
	imply that 
	the analogous estimate of Lemma \ref{lem:energy:bounds}  holds
	for the simplified Hamiltonian.
	That is 
	\begin{equation}
		\frac{1}{2} \, ( P^2  + \mu T) - C K 
		\   \leq 		 \ 
		H_{\ve,K } 
		\			   \leq  \ 
		\frac{3}{2} \,  ( P^2 + \mu T )  + C  K 
	\end{equation}
	in the sense of quadratic forms on $Q(H_K) = Q (P^2 +  \mu T )$, 
	for an approriate constant $C>0$. 
	In particular, this shows that the  intermediate dynamics is well-defined, and
	that the desired bounds hold for $\Psi_t = e^{- i t H_{\ve, K }} \vp\otimes \Omega$. 
\end{proof}

We can now combine
the energy estimates
with 
the operator bounds in Lemma \ref{lem:operator:bounds} and \ref{lemma:A}. We remind
the reader that $R_{a,b}$
is the
resolvent  defined in Equation \eqref{def:resolv}. 

\begin{corollary}
	Let $\Psi _t^\flat$ denote either 
	$  e^{- i t H_K} \varphi \otimes \Omega $
	or 
	$e^{- i t H_{\varepsilon , K  }  }    \varphi \otimes \Omega $. 
	Then, there are $  C , K_0 >0$
	such that for all
	$ \mu \geq K \ge K_0 \ge  \ve >0$
	and all $t  ,a , b \in \R $ and $ i = 1,2 $
	\begin{align}
		\label{a1}
\tag{a1}
		\|	 a (k B_{ K , i }) \Psi_t^\flat	\| 	
	 & \leq 	\	C	\	   K^{1/2} \mu^{-1 } \|  \vp\|_{H^1}  \\
	 \label{a2}
	 \tag{\rm a2}
				\|	 a (G_{a, b , i }) \Psi_t^\flat		\| 	
		& \leq 	\	C	\	 |b -a |^{1/2}  K^{1/2 }  \mu^{ -1/2 }   \|  \vp\|_{H^1}  \ . 
	\end{align}
Additionally, 
\begin{align}
	\label{A1}
	\tag{A1}
 		\|	  ( N+3)^{-1/2} A_{ K , i }  \Psi_t	^\flat 		\| 
  &  \ 			\leq     \ 
		C K^{ 1/ 2 } \mu^{-1 }  \| \vp\|_{H^1}   \\ 
		\label{A2}
		\tag{A2}
			\|	(N+3)^{-1} \,  A_{K,i } \,   R_{K ,\infty} 	\|  
  &  \ 			\leq     \ 
		C K^{  -1/2 } 	 \mu^{ -1/2 }
		 \\ 
		 \label{A3}
		 \tag{A3}
 			\|R_{K ,\infty}   \,   A_{K,i }^*  \, 	(N+3)^{-1}   \|   
  &  \ 			\leq     \ 
	C K^{  -1/2 } 	 \mu^{ -1/2 } .
\end{align}
\end{corollary}

\subsection{Removing subleading terms}
\label{subsection:removal}
 
We first remove the difference $H_K - H_{\ve,K}$ from
 the evolution. 
 This way, we get rid of the subleading terms $p\cdot a(kB_K) + a^*(kB_K) \cdot p + a^*(kB_K) a(kB_K)$,  the infrared part of the field operator $\phi (G_{\ve})$,  and the remainder potential  $   V_K$.

\begin{lemma}[Removing lemma]
	\label{lemma:removal}
	There are $C , K_0 >0$ 
such that
for all $ \mu \geq K \ge K_0 \ge  \ve > 0 $ 
	\begin{equation}
	\|   ( e^{- i t H_K }    -  e^{- i t H_{\ve,K } }  ) \varphi \otimes \Omega \|^2 
		 \leq  
 C |t| \Big(	 \,  \frac{ K	 }{\mu^{1/2} } 	 +  \ve^{1/2} K^{1/2} \, 	\Big) \, 
\| \vp \|_{H^1}^2 \ .
\end{equation}
\end{lemma}

\begin{proof}
	Let us denote in the proof 
$
		\Psi_t^{ (\ve, K )} = e^{ - i t H_{\ve,K}} ( \varphi  \otimes \Omega )  
$
and
$
\Psi_t^{(K)} = e^{- i t H_K} ( \varphi  \otimes \Omega ) $.
Then, we have thanks to Duhamel's formula 
	\begin{align}
	 	\|	 \Psi_t^{(K)}   - 		\Psi_t^{ (\ve, K )}	\|^2 
			\notag
& 				 = 
				2 \textrm{Im}
				\sum_{ i = 1 ,2 }
				\int_0^t 				 
\langle
\Psi_s^{(K)}   , 
 \big( 
  \sqrt \mu \,  p_i \cdot a (k B_{K,i }) +  \sqrt \mu  \, a^*(k B_{K,i })  \cdot p_i 
\big)  
 \Psi_{s}^{( \ve, K )}
\rangle ds  \\ 
\notag 
& \quad  + 
2 \textrm{Im}
\sum_{ i = 1 ,2 }
\int_0^t 
\langle
\Psi_s^{(K)}   , 
\mu  \, 
a^* (k B_{K, i }) a (k B_{K,i }) 
\big)   
\Psi_{s}^{( \ve, K )}
\rangle  ds  \\   
 				& \quad   +  
 2 \textrm{Im}
 				\sum_{ i = 1 ,2 }
 \int_0^t 
 \langle
 \Psi_s^{(K)}  , 
 \sqrt \mu \phi(G_{\ve, i }) 
  \Psi_{s}^{( \ve, K )} 
 \rangle ds    \ . 
	\end{align} 
For the first two lines
we use the Cauchy--Schwarz inequality and 
then    \eqref{E1} and  \eqref{a1} 
\begin{align}
\mu^{1/2 } \,  | 	\<  \Psi_s^{(K)}  ,    \, p \cdot   a (k B_K)   \   \Psi_s^{ (\ve ,K )}    \> | 
 &  	\	\leq 	 	\  C 
 \mu^{-1/2 }  K \| \vp\|_{H^1}^2 \\
\mu^{1/2 } \, 
 | 	\<  \Psi_s^{(K)}  ,  \ 
   a^*(k B_K)  \cdot p   
 	 \     \Psi_s^{ (\ve ,K )}    \> | 
 &  	\	\leq 	 	\  C 
 \mu^{-1/2 }  K \| \vp\|_{H^1}^2 \\ 
\mu \, 
 | 	\<  \Psi_s^{(K)} ,   \, 
 a^* (k B_K)  \, a (k B_K)   \  \Psi_s^{ (\ve ,K )}    \> | 
 &  	\	\leq 	 	\  C  
 \mu^{-1 }  K \| \vp\|_{H^1}^2 
\end{align}
independently of the index $i \in \{ 1,2 \}. $
 On the other hand, for the line covering the removal of the infared term
 we use  the Cauchy--Schwarz inequality
 and then 
 \eqref{a2} 
 with $a =  0 $ and $b=\ve$
\begin{align}
\mu^{1/2 } \, 
|  \langle  \Psi_s^{(K)}  ,  \phi (G_\ve )	 \Psi_s^{ (\ve, K )} \rangle  |  
& \ \leq  \ 
C \ve^{1/2 } K^{1/2} \| \vp\|_{H^1} \ .
\end{align}
 This finishes the proof
 after we use $  \mu^{-1 } \leq \mu^{-1/2}$
 and $\| \vp\|_{H^1}\leq \| \vp\|_{H^1}^2$ and gather estimates.
\end{proof}

\subsection{Main expansion}
\label{subsection:main}

We now compare the intermediate dynamics with the effective dynamics. The next statement contains the main cancellation between the two dynamics.

\begin{proposition}
	\label{thm2}
	There are constants $C ,K_0 >0$ 
so that for all $\varphi \in H^2(\mathbb R^6)$ with $\| \varphi \|_{L^2} =1$
	\begin{align}
		\notag 
 		 			\|   ( e^{- i t H_{\ve,K} }    - e^{- i t h^{\rm eff}} ) \varphi \otimes \Omega 
		  \|^2 
		&  	\leq
	C (1+|t| )  
	\|  	\vp	\|_{H^2}
	\Big(		
	\frac{ K }{\mu^{1/2}} 
	\ln( K / \ve ) 
	+
	\ve 		
	+ 
	\frac{1}{K^{1/2}}
	\Big)
	\end{align}
	for all $t\in \mathbb R$ and  $ \mu \geq K \geq K_0 \ge  \ve > 0 $. 
\end{proposition}
In order to ease the notation, we 
will drop the superscripts $(\ve, K )$
from the wave-function $\Psi_t^{(\ve,K)}$, 
and we denote by $\Phi_t$ the \textit{effective} wave-function. 
In other words  we let 
\begin{equation}\label{def:Psi:t:Phi:t}
		\Psi_t
		 \equiv e^{ - i t H_{\ve,K}} 
		( \varphi  \otimes \Omega )  
\qquad 
\t{and}
\qquad 
		\Phi_t \equiv e^{- i t \h }( \varphi  \otimes \Omega )    \ . 
\end{equation}

\begin{proof}[Proof of Proposition \ref{thm2}]

	We compare both dynamics with Duhamel's formula 
	\begin{align}
				\notag
		\|    \Psi_t
		 - \Phi_t  \|^2
		& 	 \ = \ 
		2\textrm{Im}
		\int_0^t 
		\langle  
		\Psi_s
		 \, , \, 
		\Big(
		\textstyle 
		\sum_{i=1}^2
		\,   \mu^\frac{1}{2} \, 
		a^* (G_{\ve, K,i }) 
		+  \, E_K^{(0)} \,   +   \, W(x_1-x_2) 
		\Big)
		\Phi_s 
		\rangle ds   \\
				\notag
		&  \quad    \ + \ 
		2\textrm{Im}
		\int_0^t 
		\langle  
		\Psi_s
		, 
		\  \textstyle 
		\big( 
		\sum_{i=1}^2
		\mu  \,   A_{K,i}^* 
		+ e_0 
		\big) 
		\ 
		\Phi_s 
		\rangle   ds \\
		&  \quad     \ + \ 
		2\textrm{Im}
		\int_0^t 
		\langle  
		\Psi_s
		 , 
		V_K(x_1 - x_2 ) \, 
		\Phi_s 
		\rangle  ds  \ ,
		\label{exp} 
	\end{align}
	where we introduce
\begin{align}	
E_K^{(0)} \equiv  E_K - e_0
\end{align}
with $E_K$ defined in \eqref{E}. 
First, note that  $V_K$ can be    controlled with Lemma \ref{prop:w:estimate}
 and $\| V_K\|_{L^2} \leq C K^{-1/2}$ via 
\begin{equation}
   | 			\langle  
	\Psi_s
	 , 
	V_K(x_1 - x_2 ) \, 
	\Phi_s 
	\rangle   | 
	\leq C 
	  \|   \Psi_s
	   \|
	  \|		V_K(x_1 - x_2 ) \, 
	  \Phi_s 	\| 
	  \leq 
	  C K^{- 1/2 } \|   \vp	\|_{H^2} \ ,
\end{equation}
 	In the next two subsections, we study the $a^*$ term and the $A^*$ terms separately in
		 Lemmas \ref{prop1} and \ref{prop2}, respectively.  
		 This gives the desired estimate.
		\end{proof}

\subsubsection{Expansion for the $a^*$ term}
 
We  will prove 
\begin{lemma}
	\label{prop1}
	There are $C , K_0 >0$ 
such that
for all $ \mu \geq K \geq K_0 \geq \ve > 0 $ 
and $t \in \R $
 \begin{align*}
  \bigg| \ 	\int_0^t 
 \langle 
    \Psi_s
     ,
\bigg[  
 \sum_{i=1}^2 
 \mu^{ \frac{1}{2}} a^* (G_{\ve , K ,  i }) 
  +  E_K^{(0)}   &  +  W(x_1-x_2) 
 \bigg] 
 \Phi_s 
 \rangle 
 ds  \  \bigg|   \\
 &  	\leq
 C (1+ | t |) 
 \|  	\vp	\|_{H^2}
 \Big(		
\mu^{-1/2 } K 
 \ln( K / \ve ) 
 +
 \ve 		
 + 
K^{-1/2 }
 		\Big)
 \end{align*}
\end{lemma}
	
	\begin{proof}
First, let  us fix $ i \in \{1,2 \}$. 
The proof  relies on cancellations induced by fast oscillations of the unitary $e^{-i H_{\varepsilon , K  }}$. 
To exploit these oscillations, we use integration by parts  
in the time variable 
to obtain an expansion 
for the term involving $a^*(G_{\ve K , i })$.
Namely, 
in view of  $a^* (G_{\ve, K, i  }) \Phi_s \in \t{Ran} \, Q_{\ve, K }$
we  obtain 
		\begin{align}			
	\notag
	& 	\hspace{-1cm}
	 \, 
		\int_0^t 
	 	\langle 
		\Psi_s , 
		a^* (G_{\ve , K, i}) 
		\Phi_s 
		\rangle 
		ds  \\ 
		\notag
		& 
				\hspace{-1cm}
				 =
		\, \frac{1}{i \mu }
		\int_0^t 
		\langle 
		\Psi_s,   
		e^{ - i s ( \mu T_{\ve, K }  + P^2 + E_K  )     }
		\bigg( \frac{\rm d}{{\rm d}s} e^{   i s ( \mu T_{\ve, K }  + P^2+  E_K  )     }
		R_{\ve, K } \bigg)
		a^* (G_{\ve , K, i}) 
		\Phi_s 
		\rangle 
		ds  \\  
		\notag
		&  
				\hspace{-1cm}
				= 
		\frac{1}{i\mu}
		\bigg(	
		\langle \Psi_t ,    R_{\ve,K} 
		a^* (G_{\ve , K, i}) 
		\Phi_t     \rangle  		
		-   \langle \Psi_0 ,   R_{\ve,K} a^* (G_{\ve,K  , i })  \Phi_0    \rangle  		 
		\bigg) \\
		\notag
		&  
			 	\hspace{- 0.8 cm}
			 	 - 
		\frac{1}{ \mu}
		\int_0^t 
		\langle 
		\Psi_s ,  
		\bigg[	(T - T_{\ve, K  } )  
		+
		V_K 
		+
		\textstyle 
		\sum_{ j =1}^2	 \mu^{\frac{1}{2}} \phi(G_{\ve,K , j  })	 
		+
		\mu ( A_{K, j } + A_{K,j }^*)	  		\bigg]
		R_{\ve,K} 
		a^* (G_{\ve , K, i}) 
		\Phi_s
		\rangle  ds  \\ 
		\label{eq1} 
		& 
					 	\hspace{- 0.8 cm}
		+ 	\frac{1}{ \mu}
		\int_0^t 
		\langle 
		\Psi_s , 
		R_{\ve,K}
		a^* (G_{\ve , K, i})  
		h^{\rm eff}   \Phi_s 
		\rangle    ds   \ . 
	\end{align} 
	Throughout the proof, we introduce the following
 notation combined with an estimate
 \begin{equation}
 	\label{xi}
\xi_{a,b}^{\Phi}  \equiv  R_{a,b } a^*(G_{a,b})\Phi
\quad
\t{and}
\quad 
 \|    \xi_{a,b}^{\varphi \otimes \Omega}   \| \leq C   |  \log (b/a)| \|  \varphi   \|  
 \end{equation}
 for some $C>0$ and all $  b \geq a \geq 0 $. 
 
	First, we note that 
	the first line of \eqref{eq1}
	can be estimated with \eqref{xi}
	and is   $O( \mu^{-1 } \ln (K/\ve) )$.
As for the third term, 
		  note by Hardy's inequality $\| W \varphi \| \le C \| 
P   \varphi \| $ we have 
	\begin{align}\label{eq:heff:H2}
		\| h^{\rm eff} \varphi \| \le C \|   \varphi \|_{H^2 } . 
	\end{align}
Thus, this term is $ O ( |t| \mu^{-1} \ln(K/\ve ) \|  P \vp  \|) $. 
	Additionally, letting $\ell \in \{1,2 \}$ be the complimentary 
	particle index to $i $, we have thanks to Lemma \ref{prop:w:estimate}
	and 
	$\| V_K\|_{L^2} \leq C K^{-1/2 }$
	\begin{equation}
		\mu^{-1}
 | \langle 
\Psi_s ,  
\, V_K   \, 
\xi_{\ve , K ,i }^{\Phi_s }
\rangle  | 
\leq    
 C 		\mu^{-1} 
 \|    V_K\|_{L^2} 
 \|  (1 + p_\ell)^2 \xi_{\ve , K ,i }^{\Phi_s }     \|
 \leq
\frac{C \ln (K/ \ve)}{ \mu K^{1/2 }} \| (1 + p_\ell^2)  \vp   \|_{L^2}\ . 
	\end{equation}

Second,  we note that 
the $T - T_{\ve,K }$ term and the $A_K$ term both vanish, as both of these operators commute with $R_{\ve,K} a^*(G_{\ve , K , i} )$. 
We now multiply \eqref{eq1} with $\mu^{1/2}$, 
add the $E_K^{(0)} + W(x_1-x_2)$ terms
and find that  
	\begin{align}
		\bigg|
		& \ 	\int_0^t 
		\langle 
		\Psi_s ,
		\big( 
		\textstyle 
		\sum_{i=1}^2 
		\mu^{ \frac{1}{2}} a^* (G_{\ve , K ,  i }) +  E_K + W(x_1-x_2) \big) 
		\Phi_s 
		\rangle 
		ds   
		 \ \bigg| \notag\\
		\notag
		&   
		\  \le \ 
		\bigg | \ 	\int_0^t 
		\langle 
		\Psi_s ,   
		\, 	\Big( -
		\textstyle 	 \sum_{i , j =1}^2 
		a (G_{\ve , K, j }) R_{\ve,K} a^* (G_{\ve , K, i }) 
		\  + \   E_K^{(0)}  \  +  \ W(x_1-x_2)
		\Big)  \, 
		\Phi_s 
		\rangle 
		ds  
 \ 		\Big|   \\
		&  \quad  \ + \   \bigg|  \ 
		\int_0^t 
		\langle 
		\Psi_s,   
		\textstyle 	 \sum_{i , j =1}^2 
		a^* (G_{\ve , K, j })  \,    R_{\ve,K}   \,   a^* (G_{\ve , K, i })  \,  
		\Phi_s \rangle  ds  \   \bigg|   \notag\\
		\notag
		& \quad   \  + \   \bigg|  \ 
		\mu^\frac{1}{2}
		\int_0^t 
		\langle 
		\Psi_s, 
		\textstyle 
		\sum_{i,j =1 }^2
		\, 	A_{K,j}^*  \, 
		R_{\ve,K }  \,  a^*(G_{\ve,K,i  }) \, 
		\Phi_s 
		\rangle  ds \ \bigg|  \\
		\label{eq11}
		&  \quad \ 	+  \ 	 \frac{C  \ln  (K / \ve)      }{\mu^{1/2} }   
 (1  + | t|     ) \|  \vp \|_{H^2}
	\end{align}
	where we  used  $K \geq K_0$
to absorb  all error terms into \eqref{eq11}.
Notice that a relative  sign $(-1)$  in the first term
has been introduced thanks to integration by parts. 

The first term in \eqref{eq11} is controlled in detail  in  Lemma  \ref{lemma:cancel} below.
The second term of \eqref{eq11} 
can be estimated 
with the Cauchy--Schwarz inequality, 
\eqref{a2} and \eqref{xi}   as follows 
	\begin{align}
	\notag
	\Big| 
	\langle 
	\Psi_s, 
	a^* (G_{\ve,K , j  }) 
	R_{\ve,K } a^*(G_{\ve,K, i  })
	\Phi_s 
	\rangle  
	\Big| 
	&   \ 	\leq  \ 
	\|	 a(G_{\ve,K}) \Psi_s		\| 
	\|	\xi_{\ve,K}^{\Phi_s} 		\|		
	\ \leq \ 
	C    K \mu^{-1/2}
	\ln ( K / \ve )  \|  \vp \|_{H^1 } \ . 
\end{align}
The third term of \eqref{eq11} 
can be estimated with the
Cauchy--Schwarz inequality, 
the pull-through formula for number operators,  
 \eqref{A1}
and \eqref{xi}  
\begin{align*}
\mu^\frac{1}{2}
	\Big|  
	\langle 
	\Psi_s ,  	A_{K , j }^* 
	R_{\ve,K } a^* (G_{\ve , K ,  i }) 
	\Phi_s 
	\rangle  
	\Big|    
& 	\leq 
	\mu^\frac{1}{2}
	\|	  ( N+3)^{- \frac{1}{2}} A_K  \Psi_s	 		\|
	\|	 (N+3) \xi_{\ve,K}^{\Phi_s}  	\|    \\ 
 & 	\leq 
	C
	\frac{ K^{1/2}{   \ln ( K /\ve)}}{\mu^{1/2 } } 
 \|  \vp \|_{H^1} \ . 
\end{align*} 
This finishes the proof after we collect all estimates. 
 \end{proof}

For the next lemma, 
we denote  $W_{12} \equiv W (x_1 - x_2)$ so that 
 \begin{equation} 
\label{W12}
	W_{12}  =  2 \Re \int_{\R^3 } e^{ik (x_1 - x_2)} \frac{dk}{|k|^2}
	\quad
\t{and also recall}
	\quad 
		 E_K^{(0)} =      \int_{ |k| \leq K }
 	 \frac{ 2 dk }{  |k|  (  |k | + \mu^{-1 } k^2 )		} \ . 
 \end{equation}
\begin{lemma}[Cancellation of the effective potential] 
	\label{lemma:cancel}
	There are $C ,K_0 >0$ 
such that
for all $ \mu \geq K \geq K_0 \geq \ve > 0 $ 
and $\phi \in H^2(\R^6)$
\begin{align*} 
 \Big\| 
\Big( 	 
\textstyle 
\sum_{ i  = 1,2 }
   	a   (G_{\ve,K,i }) 
		R_{\ve,K } a^*(G_{\ve,K,i })
  -   E_K^{(0)} \Big)	
 \phi \otimes \Omega  
  \Big\| 
	 & 	\,     \leq  \, 
C
		\Big( 
		 \frac{K \ln (K/\ve )}{\mu }  +        \ve 		
		 	\Big) 
			\|    \phi \|_{H^2}
		 \\
 \Big\|  
\Big(   
\textstyle 
\sum_{ i \neq j }	
a (G_{\ve,K, i  }) 
		R_{\ve,K } a^*(G_{\ve,K, j  })  
  -     W_{12} \Big)	
\phi \otimes \Omega 
   \Big\| 
	& 	\,  \leq   \, 
		C
		\Big( 
				 \frac{K \ln (K/\ve )}{\mu }  +        \ve 		
				 + \frac{1}{K^{1/2}} \bigg) 			\|    \phi \|_{H^2} \ . 
\end{align*}
\begin{proof}   
Let us take $i=1$, as the other component is analogous.  
We start with the following   computation
	which employs standard commutation relations 
	\begin{equation}
			a   (G_{\ve,K, 1 }) 
		R_{\ve,K } a^*(G_{\ve,K,1  })
		 \,  \phi  \otimes \Omega 
		 = 
		 \int_{  \ve \leq |k| \leq K  }
 \frac{ d  k }{ | k | (  | k | + \mu^{-1} ( ( p_1 - k   )^2 +p_2^2 +  E_K    )) }
 \,   \phi  \otimes \Omega 
	\end{equation}
for  any  $\phi  \in H^2 (\R^6)$.  
We proceed to  expand the denominator   as follows 
\begin{align}
	\notag 
&  \frac{1}{ |k |}	 \frac{  1 }{  | k | + \mu^{-1} (( p_1 - k   )^2 + p_2^2 + E_K   )  } 
 \  -  \ 
  \frac{1}{ |k |}
   \frac{  1 }{   | k | + \mu^{-1} |k|^2  }  \\
   \label{eq2}
& 	   
 \ 	   =   \
  \frac{1}{ |k |}
	   	 \frac{  1 }{  | k | + \mu^{-1} (( p_1 - k   )^2 + p_2^2 + E_K  )    } 
		   	 \Big(  		
		   	   \frac{p_1^2 + p_2^2 - 2p_1 \cdot k  + E_K  }{\mu} 		   	   
		   	 	\Big)	   \frac{  1 }{   | k | + \mu^{-1} |k|^2  }   \ . 
\end{align}
Let us estimate the contribution from \eqref{eq2}.
To this end, we bound each denominator by $|k|^{-1}$, 
which gives  
\begin{align}
	\notag 
	\Big\| 
	 			a   (G_{\ve,K, 1 }) 
 	 R_{\ve,K } a^*(G_{\ve,K,1  })
	 \,   \phi &  \otimes \Omega 
 \  -  \ 
	  	 \int_{  \ve \leq |k| \leq K  }
	  \frac{ d  k }{ | k | (  | k | + \mu^{-1} |k|^2    ) }
	  \,   \phi  \otimes \Omega 
	  \Big\|   \\
	  	\notag
	  & \leq 
	  \frac{1}{ \mu }
\int_{\ve \leq |k | \leq K }
 \frac{d k }{ |k |^3 }   
\Big\|
 \Big(  \, p_1^2 	\,	+ 	\,	p_2^2 	\,		+ 	\, E_K 	\,	+ \,2 |p_1 | |k| 	 \,  \Big)
\phi 
\Big\|_{L^2} \\
\label{eq4}
& \leq  
C \frac{   \ln (K/ \ve )}{\mu }
\Big(		  \|  P^2 \phi  \|_{L^2}  
+   
K  \| \phi  \|_{L^2}						\Big)
+ 
   \frac{C K}{\mu }
 \|  P  \phi  \|_{L^2}     \ . 
\end{align}
Here,  we employed 
$ \int_{\ve \leq |k| \leq K }  dk |k|^{-3} \leq C \ln(K / \ve ) $,  
$ \int_{\ve \leq |k| \leq K }   dk |k|^{-2 } \leq C K $
and
$E_K \leq    C K$.  Further note that 
\begin{equation}
	\label{eq5}
	 	  	 \int_{  \ve \leq |k| \leq K  }
	 \frac{ d  k }{ | k | (  | k | + \mu^{-1} |k|^2    ) } 
	  =  \frac{1}{2}
	  E_K^{(0)} 
	  +   O (\ve  ) \ . 
\end{equation}
 In  \eqref{eq4}, \eqref{eq5}
we
  use $ 1 \leq K  $, 
  $ 1 \leq \ln (K / \ve)   $ 
to bound all pre-factors by 
 $ K \ln ( K / \ve) \mu^{ -1 }$.
The proof is finished 
after we  sum over $ i = 1, 2$.

Let us now look at the second part. 
Observe that the different $i \neq j $ contributions are adjoints of each other. 
Thus, a  similar computation  as in the first part  shows 
	\begin{align}
\sum_{ i \neq j }	a &    (G_{\ve,K, i  }) 
 	R_{\ve,K } a^*(G_{\ve,K,  j    })
	\,  \phi  \otimes \Omega  \notag \\
& 	= 
 	\int\limits_{  \ve \leq |k| \leq K  }
2  	 \cos \big(  k(x_1 - x_2)    \big) 
	\frac{  d  k }{ | k | (  | k | + \mu^{-1} (( p_1 - k   )^2 + p_2^2 + E_K  )  ) }
	\,  \phi  \otimes \Omega  \ . 
\end{align}
We can now repeat the same denominator expansion \eqref{eq2}.
Namely, we obtain 
	\begin{align}
	\sum_{ i \neq j }	a   (G_{\ve,K, i  }) 
 	R_{\ve,K } a^*(G_{\ve,K,  j    })
	\,   \phi  \otimes \Omega   
	 \ 	=  \ 
2 \Re 
	\int_{  \ve \leq |k| \leq K  }
	\frac{   e^{i k (x_1 - x_2 )} \, d  k }{ | k | (  | k | + \mu^{-1}|k|^2    ) }
	\,   \phi  \otimes \Omega    \ + \  \mathcal E \ ,
\end{align}
where
 $\mathcal E  \in L^2(\mathbb R^6) \otimes \mathcal F $ 
 satisfies $\| \mathcal E \| \leq C K \ln (K/\ve)\mu^{-1} \|  \phi \|_{H^2}$. 
Next, we eliminate the infrared cutoff   
with the following integral estimates 
\begin{equation}
		\int_{  \ve \leq |k| \leq K  }
	\frac{  	e^{   i k(x_1 - x_2) }  \, d  k }{ | k | (  | k | + \mu^{-1}|k|^2    ) }  =
		\int_{    |k| \leq K  }
	\frac{  	e^{   i k(x_1 - x_2) }  \, d  k }{ | k | (  | k | + \mu^{-1}|k|^2    ) }
	+ O(\ve)  \ . 
\end{equation}
We now extract the leading order dependence in $\mu$ as follows. We decompose  
\begin{equation}
	\label{eq9}
			\int_{    |k| \leq K  }
	\frac{  	e^{   i k(x_1 - x_2) }  \, d  k }{ | k | (  | k | + \mu^{-1}|k|^2    ) }
	 \ =  \ 
			\int_{    |k| \leq K  }
\frac{  	e^{   i k(x_1 - x_2) }  \, d  k }{ |k|^2 }
\ + \ 
\frac{1}{\mu} \, 
			\int_{    |k| \leq K  }
\frac{  	 e^{   i k(x_1 - x_2) }    \,d  k }{  |k| + \mu^{-1} | k |^2 } \  . 
\end{equation}
Finally, 
  each term in \eqref{eq9}
  can be  bounded
with Lemma \ref{prop:w:estimate}
in combination with the    
$L^2$ bounds
\begin{align}
 \bigg\|		
\frac{\1 (|k| \leq K) }{|k|^2}
\bigg\|_{L^2 (\R^3, dk)}
\leq 
\frac{C }{K^{ 1/2 } } \ , 
\qquad 
\bigg\|	
	 \,  \frac{ \1(| k |\leq K )   	 }{ | k | + \mu^{-1} |k|^2 }
	  \, 	\bigg\|_{L^2 (\R^3, dk)}  
	  \leq 
\frac{C  K^{ 1/2 } }{\mu}
\end{align}
We now subtract the two-body potential $W_{12}$ 
with Fourier representation \eqref{W12}
and obtain
\begin{align}
	 \Big\| 
	\Big( 
	2 \Re 
	\textstyle 	\int_{    |k| \leq K  }
	e^{   i k(x_1 - x_2) }   |k|^{-2 } dk  
	- W_{12}
	\Big) 
	\phi  
	\Big\|_{L^2 }
	& \  \leq \ 
\frac{C}{K^{1/2}}  \|   \phi  \|_{H^2} \  \ ,  \\ 
\Big\|		\frac{1}{\mu} \, 
\int_{    |k| \leq K  }
\frac{  	 e^{   i k(x_1 - x_2) }    \,d  k }{  |k| + \mu^{-1} | k |^2 }  \phi  
				\Big\|_{L^2}
				& 			\	\leq  \
				\frac{ C K^{1/2} }{\mu}
				\|	 \phi 	\|_{H^2} 	 \ . 
\end{align} 
This finishes the proof
after we collect all the remainder terms. 
\end{proof}

\end{lemma}

\subsubsection{Expansion  for the $A^*$ term}

We introduce the constant
\begin{align}
	\label{e0}
e_0  = 2 \int_{\R^6} 
\frac{1}{|k_1|^3 (1 + |k_1 |)^2}
\frac{1}{|k_2 |^3 (1 + |k_2|)^2}
\frac{ | k_1 \cdot k_2|^2 }{  |k_1| + |k_2| + (k_1 + k_2)^2 				}   d k_1  d k_2  \ 
\end{align}
and prove the following statement. 
\begin{lemma}
	\label{prop2}
	There are $C , K_0 > 0$ 
such that
for all $ \mu \geq K \geq K_0  \geq \ve > 0 $ 
and $t \in \R $
	\begin{align*}
&  \Big| 
 		   	\int_0^t 
		\langle 
		\Psi_s , 			\big( 	 
		\sum_{i =1 ,2 } 
\mu  A^*_{K, i}  - e_0
\big) 
		\Phi_s 
		\rangle 
	 \, 	ds 
		\Big|  
\leq 
C (1+ |t|) 
	 \, 	\| \vp \|_{H^2} \, 
 \frac{K^{1/2}}{\mu^{1/2}}\ .
	\end{align*}
\end{lemma}

\begin{proof}[Proof]
Let $  i  \in  \{ 1 ,2 \}$. An integration-by-parts formula similar to the one used in \eqref{eq1}
implies
\begin{align}
&  	\notag
	 \, 
	\int_0^t 
	\langle 
	\Psi_s ,  
    A^*_{K , i } 	\Phi_s 
	\rangle 
	ds  \\ 
	\notag
	& = 
	\, \frac{1}{i \mu }
	\int_0^t 
	\langle 
	\Psi_s,   e^{ - i s ( \mu  T_{K,\infty  }  + p^2 +  E_K  )     }
	\bigg( \frac{\rm d}{{\rm d} s}
	 e^{   i s ( \mu T_{K,\infty  } + p^2+  E_K  )     }
	R_{\ve, K } \bigg)
    A^*_{K , i }
	\Phi_s 
	\rangle 
	ds  \\  
	& = 
	\frac{1}{i\mu}
	\bigg(	
	\langle \Psi_t ,  
	R_{K, \infty } 
    A^*_{K , i }  	   \Phi_t     \rangle  		
	-   \langle \Psi_0 ,   R_{K, \infty } 
    A^*_{K , i }  \Phi_0    \rangle  		 
	\bigg)
	+ 	\frac{1}{ \mu}
	\int_0^t 
	\langle 
	\Psi_s ,  
	R_{K , \infty }      A^*_{K , i }  	 h^{\rm eff}  \Phi_s 
	\rangle     ds 
		 	\label{A:exp}
	 \\
	&   \ \  - 
	\frac{1}{ \mu}
	\int_0^t 
	\langle 
	\Psi_s ,  
	\bigg[	(T - T_{  K, \infty   } )  +
	V_K + 
\textstyle 	\sum_{ j =1}^2  \mu^{\frac{1}{2}} \phi(G_{\ve,K , j  })	  + \mu ( A_{K, j } + A_{K, j}^*)	  		\bigg]
	R_{ K  , \infty   }  
	    A^*_{K , i } 
	  \Phi_s 
	\rangle ds \notag 
\end{align}
and we proceed by estimating each term separately.

First, we observe that the first three terms in \eqref{A:exp} can be estimated
using  Eq.   \eqref{A3} 
\begin{equation}
	\Big| 	\langle \Psi_t ,  
	R_{K, \infty } 
	A^*_{K , i }  	   \Phi_t     \rangle  		
		\Big| \leq 
		\frac{C}{K^{\frac{1}{2}}\mu^\frac{1}{2}} \ ,  \ 
	\qquad 
	\Big| 		\int_0^t 
	\langle 
	\Psi_s ,  
	R_{K , \infty }      A^*_{K , i }  	 h^{\rm eff}  \Phi_s 
	\rangle  ds 	\Big|  
	\leq 
	\frac{C|t|}{K^{\frac{1}{2}}\mu^\frac{1}{2}}
	\| \vp \|_{H^2}  
\end{equation}
where we also used $\|   h^{\rm eff} \vp   \| \leq  C \|  \vp  \|_{H^2}.  $
Additionally, the potential   $V_K$
can be estimated with  Lemma \ref{prop:w:estimate}. That is, 
letting $ \ell \in \{ 1,2\}$ be the complementary index to $i $, we have
\begin{equation}
	 	\Big| 	\langle \Psi_t ,  
V_ K 
	 R_{K, \infty } 
	 A^*_{K , i }  	   \Phi_t     \rangle  		
	 \Big|
	  \leq 
C	  \| V_K\|_{L^2}
\|  	 R_{K, \infty } 
A^*_{K , i }  	  (1 + p_\ell^2)   \Phi_t     \|
\leq 
	 \frac{C}{K \mu^\frac{1}{2}}  \|   \vp \|_{H^2}
\end{equation}

Secondly, we now turn to study the remainding terms in the second line of \eqref{A:exp}. 
Noting that the contribution $ T - T_{K,\infty}$ vanishes, we turn to study
the following four contributions, 
which we obtain after summing over $ i \in \{ 1,2 \}$
\begin{align}
	\notag 
	\A (t)  &   \equiv 
		 {      \sum_{  i, j =1}^2   	\int_0^t  	} 
		  \ 	 \mu^{\frac{1}{2}}  \ 
	\langle 
	\Psi_s ,  
 	\, a(G_{\ve,K , j  })	 \, 
	R_{ K  , \infty   }   \, 
	A^*_{K , i } 
	\Phi_s  
	\rangle \   ds   \\
	\notag 
	& \quad   \,   +  
	 {      \sum_{ i,  j =1}^2   	\int_0^t  	} 
	  \  	 \mu^{\frac{1}{2}}  \ 
\langle 
\Psi_s ,   	\, a^* (G_{\ve,K , j  })	 \, 
R_{ K  , \infty   }   \, 
A^*_{K , i } 
\Phi_s  
\rangle \   ds   \\
\notag 
	&  \quad    \,  + 
 {      \sum_{  i , j =1}^2   	\int_0^t  	} 
\ 	 \mu  \ 
\langle 
\Psi_s ,   
A_{K, j} \, 
R_{ K  , \infty   }   \, 
A^*_{K , i } 
\Phi_s  
\rangle  \  ds   \\
	\label{A}
	&  \quad    \,   +
 {      \sum_{ i,  j =1}^2   	\int_0^t  	} 
 \ 	 \mu  \ 
\langle 
\Psi_s ,   
A^*_{K, j} \, 
R_{ K  , \infty   }   \, 
A^*_{K , i } 
\Phi_s  
\rangle  \  ds  \ . 
\end{align}
Note that the  first  term in \eqref{A} is zero, 
as $ a (G_{\ve, K ,j})$ commutes with $A^*_{K,i}$  and $ a (G_{\ve, K ,j}) \Phi_s = 0$. The second term in \eqref{A}
can be estimated with  
\eqref{a2}
and
\eqref{A3}
	\begin{align*} 
	\Big|   
	\, 	\langle 
	\Psi_s , 		a^* (G_{\ve,K })
	R_{ K  , \infty   }  A_K^*  \Phi_s  
	\rangle \,  
	\Big|       
	\ \leq \   2 \, 
	\|					a(G_{\ve,K }) \Psi_s		\|   \, 
	\|    R_{K, \infty} A_{K}^* ( N +2)^{ -1 } \|  
	& \  \leq 		 \ 
	C   K^{1/2} \mu^{ -  1  }
\end{align*}
The third term in \eqref{A} together with the self-energy contribution $e_0$ is estimated in detail in
Lemma \ref{prop3} below.
Finally, the fourth term in \eqref{A}
can be estimated 
with  \eqref{A1} and \eqref{A3}  
\begin{align*}
	\big|  	\langle   \Psi_s, 
	A_K^* R_{K, \infty} A_K^* \Phi_s   \rangle 
	\big|    \ 	\leq  \ 
	C \, 	\|	 ( N  +2)^{- \frac{1}{2}} 
	A_K \Psi_s 	\|
	\|	 R_{ K , \infty} 	A_K^* \Phi_s 	\|  
	\   \leq  \ 
	C  \mu^{ - 3/2} \ . 
\end{align*}  
This finishes the proof  after we account for the 
additional $\mu^{1/2}$
and $\mu$ factors from \eqref{A}.
\end{proof}

Let us now turn to study the third term of \eqref{A}, which gives rise to the self-energy contribution $e_0$. After subtracting $e_0$, we denote this term by $S$ and decompose it into
\begin{align} 
S 
	& = 
	\int_0^t  \langle 
	\Psi_s ,
\Big(    \sum_{i=1,2} 
 \mu  A_{K,i}   R_{ K  , \infty   }  A_{K,i}^*  - e_0 \Big) 	  \Phi_s
	\rangle  ds \notag\\
	& \qquad  \qquad + 
	\int_0^t 
	\mu \langle 
	\Psi_s ,
	 ( A_{K,1}   R_{ K  , \infty   }  A_{K,2}^* + 
	 {\rm h.c. } ) 	\Phi_s
	\rangle ds  ,
	\label{eq:new:scalar}
\end{align}
where we call the first integral $S_{\rm diag}$, 
 and the second integral $S_{\rm off}$. 
We prove the following estimate.

\begin{lemma} 
	\label{prop3}
There is $C , K_0>0$
such that for all  $\mu \geq K \geq K_0 \geq \ve >0$
and $ t \in \R $
	\begin{equation}
		|S| \leq 
\frac{C|t | K }{\mu}
 \|  \vp \|_{H^2} \ . 
	\end{equation}
\end{lemma}

\begin{proof}
For the proof, we estimate the off-diagonal and diagonal terms separetely. 
The proof of the lemma follows
from the estimates \eqref{S:off} and \eqref{S:on}
combined with 
Hardy's inequality and energy conservation for
$\phi = e^{- i s h^{\rm eff}}\vp $, i.e. 
$\|   \phi \|_{H^2} 
\le
C  \| h^{\rm eff} \phi \|_{L^2} 
=
C
 \| h^{\rm eff} \varphi \|_{L^2} \leq C  \| \vp\|_{H^2}$.
In what follows, we will be using the notation
\begin{equation}
	F(k) =  
	\frac{\textbf{1}_{|k | \geq K }   \, k }{ |k|^{ \frac{1}{2}}  ( \mu |k| + k^2) } \in \R^3 \  ,  \qquad k \in \R^3 \ . 
\end{equation}
Note that $k B_{K, i } (k) = e^{-ik x_i }F(k)$. 

\smallskip 

 {\textit{The off-diagonal terms.}}
Let $\phi \in H^2(\R^6)$.
Using the pull-through formula, 
we compute 
 \begin{align}
\label{trace}
  	 A_{K,1}   R_{K , \infty} A_{K,2}^*
  	  \   \phi \otimes \Omega 
  	   \notag  
	\notag
	&   = 
	\int 
	e^{i (x_1 - x_2) (k_1 + k_2)}  
	\frac{    2 \,   | F(k_1) \cdot F(k_2) |^2   }{ 
		|k_1|
		+
				|k_2| +  \frac{ E_K  + p_1^2  +  (p_2 - k_1 - k_2 )^2  }{\mu } }
	\, d k_1 d k_2  		   \phi \otimes \Omega 
			 \\
	 & 
	\equiv 
	 \int 
 e^{i (x_1 - x_2) (k_1 + k_2)}  g( k_1 , k_2)  d k_1 d k_2  \ \   \phi \otimes \Omega   \  , 
\end{align}
where    $g(k_1 , k_2)$ is understood as  
a bounded operator on $L^2(\R^6)$, 
and we note that the state $A_{K,1}   R_{K , \infty} A_{K,2}^*\phi \otimes \Omega $ still belongs to $\text{ker}(T)$, i.e. no real bosons are created.
Next, observe that by comparison of the denominators 
\begin{align}
	g(k_1, k_2) 
	 \    -  \ 
	     \frac{ 2  |   F (k_1) \cdot F(k_2)   |^2 }{
	   		|k_1|
	   	+
	   	|k_2|  + \frac{(k_1 + k_2)^2}{\mu} }
 \ 	   =  \ 
 g(k_1  , k_2)  	
 \bigg(  
  \frac{
	 		 E_K  + p_1^2 +p_2^2  - 2 p_2 \cdot (k_1+k_2)   }{ \mu 	|k_1|
	 		 +
	 		 \mu  |k_2|  +  {(k_1 + k_2)^2}  } 		\bigg)  \ . 
\end{align}
Let us note that  $g(k_1, k_2)$  is supported
on a set where  both $|k_1| \geq K \geq K_0 $ 
and $|k_2| \geq K \geq K_0 $. 
Additionally, note   $ 0 < E_K  \leq C K $.
Hence, we obtain the estimate
\begin{align}
	 \bigg\|
	 \Big(  
	 	g(k_1, k_2) 
	 - 
	 \frac{ 2  |   F (k_1) \cdot F(k_2)   |^2 }{
	 	|k_1|
	 	+
	 	|k_2|  + \frac{(k_1 + k_2)^2}{\mu} }
 	\Big) \phi \otimes \Omega 
	 \bigg\|
	 \leq 
	 \frac{C  K }{\mu} \|	 g(k_1 , k_2)		\| \|  \phi \|_{H^2} \ . 
\end{align}
On the other hand, thanks to 
$	|k_1|
+
|k_2|   \ge |k_1|^{1/2} |k_2|^{1/2 }$, 
we have the following
inequality 
\begin{align}\label{eq:scalar:term:integral}
	\int _{\R^6}  \| g(k_1,k_2)  \| d k_1  d k_2 
	\  \le  \ 
\frac{1}{\mu} \, 
	\bigg[   \int_{|k |\geq \frac{K}{\mu }}     \frac{ |k|^{1/2} d k }{( k^2 + |k|)^2}  \bigg]^2 
	\le \frac{C }{\mu}
\end{align}
Multiplying by $\mu$, this leads to 
\begin{align}
\label{eq12}
\bigg\|   
\bigg(  	 \mu  \,   	 A_{K,1}   R_{K , \infty} A_{K,2}^*
	 -  
	  \int _{\R^6 }
	 \frac{	    e^{i (x_1 - x_2) (k_1 + k_2)}     |   F (k_1) \cdot F(k_2)   |^2 }{	 \mu |k_1|
	 	+
	 \mu 	|k_2|  +  {(k_1 + k_2)^2} }  d k_1 d k_2  
 \bigg)  
  \phi  \otimes \Omega  
	 \bigg\| 
	\leq 
	\frac{C  K }{\mu} \|  \phi\|_{H^2} \ . 
\end{align}
Let us now    introduce the following auxiliary potential, and use the change 
of variables $k_i \mapsto \mu k_i$  
\begin{align}
	\notag 
	\widetilde W(x) 
	  &  \, 
	  \equiv  \,     
	 	    	 \int_{\R^6} 
	 \frac{	    e^{i x   (k_1 + k_2)}     |   F (k_1) \cdot F(k_2)   |^2 }{ 
	  \mu |k_1|
	 +
	 \mu 	|k_2|  +  (k_1 + k_2)^2	  	 }  d k_1 d k_2  \\
 \notag 
	  &  \, = \,     
	 	    	 \int_{  | k_1| , |k_2|  \geq  K / \mu 		} 
e^{ i  \mu x (  k_1 + k_2  ) }
  	w( k_1 , k_2) d k_1 d k_2  \\
  	 &  \, = \,     
  	\int_{\R^6} 
  	e^{ i  \mu x (  k_1 + k_2  ) }
  	w( k_1 , k_2) d k_1 d k_2 
  	+ O( K /\mu )^2 
  	\label{eq10}
\end{align}
where
\begin{equation}
 w( k_1 , k_2) 
 \equiv 
 \frac{1}{| k_1|^3 (1 + | k_1 |)^2}
  \frac{1}{| k_2 |^3 (1 + | k_2|)^2}
   \frac{  (  k_1 \cdot k_2)^2 }{  | k_1| + | k_2 | + ( k_1 + k_2)^2 				} \  , 
\end{equation}
The $O(K/\mu)^2$ term in \eqref{eq10} comes from the small integrals $0 \leq |k_i | \leq \frac{K}{\mu}$ and is uniform in $x \in \R^3$.
Next, we exploit 
 the rapidly oscillating phase in the 
exponential. 
 That is, for $x \neq 0 $ we get
that for each vector $a \in \R^3 $
not orthogonal to $x$
\begin{align}
	\widetilde W(x) = \frac{  i }{  \mu ( a \cdot x ) } \int_{\R^6 }
	e^{i \mu x ( k_1 + k_2)} 
(	a \cdot \nabla_{k_1} ) w ( k_1 , k_2) d k_1  d k_2  
  	+ O( K /\mu )^2  
\end{align}
One may verify that\footnote{Indeed, a somewhat
	lengthy  calculation shows   
$
	|\nabla_{ k_1} w  (k_1 , k_2)		|		 
	\leq 
	8 |k_1|^{-  \frac{5}{2 }} |k_2|^{- \frac{3}{2}}
	(1 + | k_1 | )^{-2} 
	(1 + | k_2 | )^{-2}  
$
.}
  $	\|		 \nabla_{ k_1 } w 	\|_{L^1(\R^6)} <\infty  $. 
Thus,   choosing $a =x $ we obtain 
\begin{equation}
 |\widetilde W (x)| \leq \frac{C}{\mu |x|} + \frac{C  K}{\mu } \ , \qquad |x |> 0 \ . 
\end{equation}
Finally,  the potential $|x|^{-1}$ can be controlled 
with the $H^1 $ norm thanks to Hardy's inequality.
We arrive at $\|  \widetilde W (x_1  - x_2) \phi \|_{L^2} \leq C \frac{K}{\mu} \|  \phi\|_{H^1}$. 
Combined with \eqref{eq12} we obtain 
\begin{equation}\label{S:off}
 \|	 \mu \,   A_{K,1} R_{K, \infty} A_{K,2}^*
	\   \phi   \otimes \Omega  \| 
	\leq \frac{C K }{\mu} \| \phi \|_{H^2} \ .
\end{equation}
The same estimate holds if we reverse the roles of $ i= 1 $ and $ i =2 $. 

 \textit{{The diagonal terms.}}
Since the analysis is analogous to the off-diagonal terms  (but simpler), we provide only some of the details. 
We keep the same notations for $F(k)$, $g(k_1,k_2)$
and $w (k_1 , k_2) $. 
Similarly as in \eqref{trace},  
we compute for  $ \phi \in H^2(\R^6)$
 \begin{align}
  A_{K,1 }  R_{K , \infty} A_{K,1 }^*    \  
   \phi \otimes\Omega 
 	  	\notag
         = 
         \int_{\R^6}  g(k_1, k_2) d k_1 d k_2 
         \     \phi \otimes\Omega 
 \end{align}
Note that this is still a bounded operator on $L^2(\mathbb R^6)$. 
We can repeat the same argument 
leading to \eqref{eq12}
and obtain
\begin{equation}
	\bigg\|   
 \bigg(  
  \mu 	A_{K, 1} R_{K , \infty} A_{K,1}^* 
 - 
	         \int    
	 \frac{ 2  |F (k_1)|^2 |F(k_2)|^2  \, d k_1 d k_2   }{ 	 \mu |k_1|
	 	+
	  \mu 	|k_2|  +   (k_1 + k_2)^2  } 
  \bigg)  
	 \phi  \otimes \Omega  \bigg\|  
\leq  \frac{C K }{\mu} 	\|  \phi \|_{H^2}
\end{equation}
We  re-scale $k_i = \mu^{-1} k_i $
and use a simple integral estimate
to obtain
\begin{equation}
	     \int  _{\R^6}
	\frac{ 2  |F (k_1)|^2 |F(k_2)|^2  \, d k_1 d k_2   }{ 	
		\mu |k_1|
		+
	 \mu 	|k_2| +    (k_1 + k_2)^2  }  
	 = 
	 \int_{\R^6} w( k_1 , k_2) d k_1 d k_2 
	 + 
	 O(K/ \mu)^2 \ . 
\end{equation}
Note that the first term on the right-hand side 
coincides with the energy term $\frac{1}{2}e_0$ defined in \eqref{e0}.
Thus, we arrive at 
\begin{align}
	\label{S:on}
	 	\bigg\|   
	 \bigg(  
	 \mu 	A_{K, 1} R_{K , \infty} A_{K,1}^* 
	 - 
	 \frac{e_0}{2}
	 \bigg)  
	 \phi  \otimes \Omega  \bigg\|  
	 \leq  \frac{C K }{\mu} 	\|  \phi \|_{H^2}
\end{align}
and the same bound holds if we exchange    $i =1 $ with $i=2$. 
\end{proof} 
 
%

\appendix

\section{Proof of Proposition \ref{prop:ren:Nelson} and Lemma \ref{lemma:dressed:Nelson}}

The aim of this appendix is to provide the main steps in the proofs of Proposition \ref{prop:ren:Nelson} and Lemma \ref{lemma:dressed:Nelson}. 
Let us recall that $E_\Lambda  = E_\Lambda^{(0)} + e_0$
where $e_0$ is a numerical constant defined in \eqref{e0}.  
For $f, g \in L^{\infty} L^2$
and $X \in \R^6 $
we abuse notation and denote
$\< f, g\> \equiv    \int_{\R^3} \overline f ( X, k ) g(X, k ) dk $
as well as $\| f   \|_{L^2}^2  \equiv  \<  f, f\>$. 
That is, we omit the $X$ dependence in the notation.

First, we introduce the version of \eqref{eq:Gross:trafo}
	with a  UV cutoff $\Lambda \geq K>0 $ 
	\begin{align}
		\notag 
		U_{K,\Lambda} = \exp\Big(  \mu^{1/2} \sum_{i=1,2 } (  a^*(B_{K,\Lambda,i}) - a( B_{K,\Lambda,i})) \Big)    \quad \text{with} \quad B_{K,\Lambda,i }(k) = B_{K,i}(k) \1_{|k|\le \Lambda} \ . 
	\end{align}
We ease the 
	notation and write $ U \equiv U_{K, \Lambda}$
	and $B_i   \equiv B_{ K, \Lambda, i} $.  
	Following  \cite{Griesemer18} 
	we compute the action of $U $
	on each term of the Hamiltonian $H^{\rm N}_\Lambda$ defined in \eqref{eq:HN}
	with the aid of the   relations  
	\begin{align}
	\notag 	U    \  p^2_i   \   U ^* & 
		\ =  \  p_i^2 
		+ 
		\sqrt \mu     \, 
		\big[ 
		2 a^*(k B_{  i } ) \cdot  p_i
		+ 
		2 p_i  \cdot a( k B_i ) 
		-
		\phi(k^2  B_{  i } ) 
		\big] \\
		\label{rel1}
		& \quad  + 
		\mu   \,  
\big[
	2   a^*(k B_{  i})
a(k B_{K   i})
+     a(k B_{  i})^2
+ 
a^*(k B_{  i})^2 + 		    \| 	k B_{  i  }  \|^2_{L^2}
\big]
   \\[1mm]
   \label{rel2}
		U   \ d\Gamma(| k |) \   U ^* 
		& \  = \ 
		d\Gamma(| k |) -  \sqrt \mu    \, 
		\phi( | k | ( B_{ 1} +B_2 )  ) 
		+  \mu  
		\|  | k |^{1/2}  ( B_{ 1 }  + B_{2})  \|^2_{L^2}      \\[1mm]
		U  \   \phi (G_{\Lambda,i})  \ U^* & 
		\ = \ 
		\phi(G_{\Lambda,i} ) -  
		2 \, 
		\Re \,     \sqrt \mu  \, 
		\<  G_{\Lambda  , i } ,  (		 B_{1  }	 + B_{2 })		 \>  \label{rel3}
	\end{align}
where we used 
$  \Re \< B_1 + B_2 , - i \nabla_1 (B_1 + B_2 ) \>  =0 $. 
	We can now collect 
	all the scalar contributions
	\begin{align*}
		\mu  	\textstyle 
		\sum_{i=1,2}  		    \| k B_{  i} \|^2_{L^2}  
		 \ + \ 
		 \mu^2  \, 
		 \| | k |^{ \frac{1}{2}} 
		& 	(B_{ 1}
		+   B_{ 2} )   \|^2_{L^2}    
		\ - \  	 2 \, \mu  \,  \textstyle \sum_{i = 1,2 }
		\Re \,      
		\<  G_{\Lambda  , i } ,  (	B_1 + B_2 )		 \>     \\
		&  \ =  \ 
		 E_K^{(0)}  \,  - \,  E^{(0)}_\Lambda  
		 \ + \  2 \mu  \,  \Re 
		\Big[ \mu \<  B_1 , | k | B_2\> 
		- 2 \<B_1, G_{\Lambda, 2}\> 
		\Big]  \ . 
	\end{align*} 
	We can now rewrite
	\begin{align}
		\notag 
		2   \mu   \, \Re 
		\Big[ \mu \<  B_1 , |  k | B_2\> 
		-2    \<B_1, G_{\Lambda, 2}\> 
		\Big] 
		&   \,  =  \, 
	 -2   \mu \,  \Re 
		\int_{  K \leq |k | \leq \Lambda		}
		\frac{  e^{ik (x_1 - x_2 )}( \mu | k | +2 k^2 ) \, d k }{ | k |( \mu | k | +k^2)^2} \\
		& \,  \equiv \,   V_{K,\Lambda} (x_1 - x_2 ) \ . 
	\end{align}
	One may verify that $\|  V_{K  , \Lambda}\|_{L^2} \leq  C K^{-1/2} $ and 
		$\|	 V_{K} - V_{K , \Lambda}	\|_{L^2} \leq   C \Lambda^{ - 1/2 } . $
	In particular, 
	thanks to Lemma \ref{prop:w:estimate}
	we have for some $C>0$ and all $ K, \Lambda>0$
\begin{align}
	\label{VK1}
 |  \< \Psi, V_{ K, \Lambda } (x_1 - x_2) \Psi\> | 
&  \leq C K^{ - 1/2 }   \< \Psi,  (1 + P^2)\Psi \>  \\
\label{VK2}
  |  \< \Psi,  ( V_{ K, \Lambda }  - V_K )  (x_1 - x_2) \Psi\> | 
&  \leq C \Lambda^{ - 1/2 }   \< \Psi,  (1 + P^2)\Psi \> \ . 
\end{align}
	On the other hand, for the
	linear terms in creation and annihilation operators we find
	\begin{equation}
		\phi ( G_{\Lambda, i}	) 
		-  \phi  (  (k^2 + \mu | k |) B_i		)  
		= 
		\phi ( G_{K , i}	 )  \ ,  \qquad  i \in \{ 1 , 2\} \ . 
	\end{equation}
	We now put everything together to see that 
	\begin{align}
		& H_{K,\Lambda}  \equiv U_{K,\Lambda} ( H_\Lambda +  E_\Lambda ) (U_{K,\Lambda})^*  \notag \\[3mm] 
		&\ \ = P^2 +  \mu T  + \sqrt \mu \sum_{i=1 ,2 }
		 \Big( \phi(G_{K,i}) + 2 p_i  \cdot a( k B_{K,\Lambda,i})  + 2 a^* ( k B_{K,\Lambda,i} )  \cdot p_i \Big)  \notag\\
		&  \quad \   +    
	\sum_{ i =1 ,2 }  \Big(   
	 2 a^*(	 k  B_{ K , \Lambda, i }	)
	 a(	 k  B_{ K, \Lambda, i }	)
	 + 
	 	 a(	 k  B_{ K, \Lambda, i }	)^2 
	 	 + 
	 	 	 a^*(	 k  B_{ k, \Lambda, i }	)^2
	  \Big)  + E_K + 
	  V_{ K, \Lambda }  
 	\label{HKL}
	\end{align}
 where we used $E_K^{(0)} - E_\Lambda^{(0)} = E_K - E_\Lambda$ and recall $P^2=p_1^2 +  p_2^2$.
 Following   \cite[Lemma 3.1]{Griesemer18} 
	we use the bounds \eqref{op1}--\eqref{op4}
	to     prove that for any $ \epsilon >0 $   there is 
  $ K_0 $   large enough
such that $\pm (H_{K,\Lambda} - P^2 - \mu T) \le \epsilon (P^2 + \mu T) + C_\epsilon$ for all $K\ge K_0$.
The new term $V_{K,\Lambda}$ satisfies the bound thanks to  \eqref{VK1}.
	We observe that  a detailed calculation in fact shows that 
	$K_0$ can be taken independent of $\mu>0$.
This now  implies that  for all $ \infty > \Lambda \ge K \geq K_0 $
the operator
$H_{K,\Lambda}$ defines a semi-bounded operator  with quadratic form domain $Q(  H_{K,\Lambda}) = Q (P^2  +  \mu T )$.
Following the proof of
\cite[Theorem 3.1]{Griesemer18} --  and equipped with the additional bound \eqref{VK2} -- we   prove that there exists a semi-bounded operator  $H_{K , \infty }$
with form domain $Q (P^2  +  \mu T ) $
such 
that  $e^{-i t H_{K,\Lambda}} \xrightarrow{\Lambda \to \infty}e^{-it  H_{K,\infty}}$ in the strong sense for every $t \in \mathbb R $. 
This concludes the proof of 
Proposition \ref{prop:ren:Nelson}.
Additionally, 
the argument proves   
the claim of  Lemma \ref{lemma:dressed:Nelson}  for 
$H^{\rm N} = (U_{K,\infty})^* H_{K,\infty} U_{K,\infty}$ provided $ K \geq K_0 $.

Let us now extend the statement of Lemma \ref{lemma:dressed:Nelson} to all $   K >  0 $.
First, we start from the representation
$H^{\rm N} = (U_{K,\infty} )^* H_{K,\infty} U_{K,\infty}$ for     $K \geq K_0  $. Observe that for all $0  <  M \le K \le L \le \infty$, we have  $U_{M,K} U_{K,L} = U_{M,L}$
which  follows  from the definition of $B$ (since the generators of $U$ on the left side commute and add up to the generator on the right side). Thus, 
\begin{align}
H^{\rm N} & = ( U_{M,K} U_{K,\infty} )^*  ( U_{M,K} H_{K,\infty} U_{M,K}^*) (U_{M,K} U_{K,\infty}) \notag \\[1mm]
& = (U_{M,\infty} )^*  ( U_{M,K} H_{K,\infty} U_{M,K}^* ) U_{M,\infty} \notag \\[1mm]
& = (U_{M,\infty} )^*  H_{M,\infty}  U_{M,\infty}. 
\end{align}
The last step can be computed explicitly using   \eqref{rel1}--\eqref{rel3} --   all the creation and  annihilation operators in $H_{K,\infty}$, except for $\phi(G_K)$,  commute with $U_{M,K}$ due to the disjoint support in $k$. This gives the claimed representation $H^{\rm N} = (U_{M})^* H_{M} U_{M}$ for any $M>  0$.  To complete the proof, it remains to show   $Q(H_M) = Q(P^2+T)$. 
Since $Q(H_K) = Q(P^2+T)$  for     $K_0  \leq K < \infty   $, 
 it is sufficient to show that $Q( U_{M,K} (P^2+T) U^*_{M,K} )  = Q(P^2+T) $.
To this end,   note that  the relations  \eqref{rel1}--\eqref{rel2}
 and the form bounds \eqref{op1}--\eqref{op4}  imply 
\begin{align}
	\label{UpU}
U_{M,K} (P^2+T) U^*_{M,K} \le C (P^2 +T +\id )
\end{align}
for some $C>0$ (possibly $K$ and $\mu$ dependent), which implies $Q( U_{M,K} (P^2+T) U^*_{M,K} )  \supset Q(P^2+T)$.
Finally,  we may reverse     the sign of $B$,  which turns $U$ into $U^*$.  Hence, we also obtain 
$Q( U_{M,K} (P^2+T) U^*_{M,K} )  \subset Q(P^2+T)$. This finishes the proof of Lemma \ref{lemma:dressed:Nelson}.

\section*{Acknowledgements}
D.M. thanks Nata\v sa Pavlovi\'c for the invitation to the University of Texas at Austin and for the hospitality offered by the department, where part of this work was performed.
E.C. gratefully acknowledges support from   NSF under grants No. DMS-2009549 and DMS-2052789 through 
Nata\v sa Pavlovi\'c. 
 \end{spacing}

\end{document}